\documentclass[lettersize,journal]{IEEEtran}
\usepackage{amsmath,amsfonts}
\usepackage{algorithmic}
\usepackage{array}
\usepackage[caption=false,font=normalsize,labelfont=sf,textfont=sf]{subfig}
\usepackage{textcomp}
\usepackage{stfloats}
\usepackage{url}
\usepackage{verbatim}
\usepackage{graphicx}
\usepackage{cite}
\usepackage{amsmath,amsthm,amssymb}
\usepackage{algorithm2e}
\usepackage{pst-node}
\usepackage{caption}
\usepackage{xcolor}
\usepackage{bbm}
\usepackage{booktabs}
\usepackage{multirow}
\usepackage{color, colortbl}
\usepackage{hyperref}
\SetKwComment{Comment}{/* }{ */}
\SetKw{KwBy}{by}

\newcommand{\scheme}{PriRoAgg\xspace}
\newtheorem{mydef}{Definition}

\newtheorem{mythm}{Theorem}

\begin{document}

\title{PriRoAgg: Achieving Robust Model Aggregation with Minimum Privacy Leakage for Federated Learning}

\author{Sizai Hou,
Songze Li,
Tayyebeh Jahani-Nezhad, and
Giuseppe Caire
\thanks{
Sizai Hou is with Hong Kong University of Science and Technology, Hong Kong, China (e-mail: shouac@connect.ust.hk).

Songze Li is with the School of Cyber Science and Engineering, Southeast University, Nanjing, China; Engineering Research Center of Blockchain Application, Supervision and Management (Southeast University), Ministry of Education (e-mail: songzeli@seu.edu.cn).

Tayyebeh Jahani-Nezhad and Giuseppe Caire are with the Department of Electrical Engineering and Computer Science, Technische Universität Berlin, 10587 Berlin, Germany (e-mail: \{t.jahani.nezhad, caire\}@tu-berlin.de).
}%

\thanks{Manuscript received July 09, 2024; revised Jan 29, 2025; accepted May 27, 2025.}%

}



\maketitle

\begin{abstract}
Federated learning (FL), as a promising machine learning paradigm for large-scale distributed data, faces two security challenges of privacy and robustness: the transmitted model updates potentially leak sensitive user information, and the lack of central control over local model updates leaves the global model susceptible to malicious attacks. 
Current solutions attempting to address both problems under the one-server FL setting fall short in the following aspects: 1) design for simple validity checks that are insufficient against advanced attacks (e.g., checking norm of individual update); and 2) have partial privacy leakage for more complicated robust aggregation algorithms (e.g., distances between model updates are leaked for multi-Krum).
In this work, we formalize a novel security notion of \textit{aggregated privacy} that characterizes the minimum amount of user information, in the form of aggregated statistics of users' updates, that is necessary to be revealed to accomplish more advanced robust aggregation.
We develop a general framework \scheme, utilizing Lagrange coded computing and distributed zero-knowledge proof, to execute a wide range of robust aggregation algorithms while satisfying aggregated privacy.
As concrete instantiations of \scheme, we construct two secure and robust protocols based on state-of-the-art robust algorithms, for which we provide full theoretical analyses on security and complexity. Extensive experiments are conducted for these protocols, demonstrating their robustness against various model integrity attacks, and their efficiency advantages over baselines.
\end{abstract}


\section{Introduction}

Federated learning (FL) is a trending collaborative machine learning (ML)  paradigm, where numerous devices collaboratively train a model on their local datasets \cite{mcmahan2017communication}. 
FL was designed to protect user data privacy, as the local model training is performed on each user's private data, and only the model updates are uploaded to the server for aggregation, hence, the server does not directly access users' data.
However, some recent research works have revealed that 
the uploaded model updates may leak users' data via attacks such as gradient inversion~\cite{zhu2019deep, geiping2020inverting, yin2021see, lyu2020survey}. 
On the other hand, due to FL's distributed nature, the global model's integrity may also be compromised by malicious users manipulating their local models.
A variety of works have proposed attacks that can degenerate global model performance, or implement backdoors to cause misclassifications \cite{rigaki2023survey, barreno2010security, xue2020machine}. 

Most of the existing works focus on developing protocols that solve either the privacy problem or the model integrity problem, whereas there is a scarcity of studies that simultaneously address both problems. The main reason for that is the conflicting nature of the two goals: robust aggregation protocols for model integrity require the central server to aggregate based on the distribution of user updates ~\cite{ghosh2020communication, karimireddy2020byzantine, velicheti2021secure, orekondy2019prediction, normthreshold, krum, guerraoui2018hidden, chen2021pois, zhao2020pdgan}, while secure aggregation protocols encrypt or mask the user updates so that they are indistinguishable statistically \cite{wei2020federated, bell2020secure, truex2020ldp, naseri2020local, el2022differential, so2022lightsecagg, jahani2023swiftagg}. 
To compatibly provide security in robust aggregation protocols, commonly used techniques in secure aggregation face many limitations that hinder the combination of secure techniques and robust protocols. For example, homomorphic encryption suffers from high complexity \cite{fang2021privacy, chen2021homomorphic, ma2022privacy, xia2024byzantine} that results in impractical implementations. Differential privacy is inherently imprecise and its masking renders direct computation of privacy data statistics infeasible ~\cite{lyu2022privacy, el2022differential} . 
A few works provide robustness and security simultaneously, but they are designed for rather simple
defensive mechanisms, such as cosine similarity \cite{chowdhury2021eiffel} and norm bounding \cite{lycklama2023rofl}.  More importantly, their detection of malicious behaviors relies on single user's self-check instead of finding overall statistics. 
In works like \cite{brea, jahani2023byzantine} that attempt to conceal the local updates from the server, there is additional privacy leakage from the intermediate computation results. Some other privacy-preserving protocols enable the execution of robust algorithms but assume stronger settings, such as non-colluding multi-server scenarios \cite{ma2022shieldfl, rathee2023elsa, li2024efficiently}.
Therefore, it remains a challenging problem to strengthen powerful robust algorithms with provable privacy under the typical single-server FL paradigm. 

To this end, this paper aims to develop a secure aggregation framework, with minimal privacy leakage, for a general class of robust model aggregation algorithms. To rigorously define information leakage including the final model aggregation revealed to the server, 
we first introduce a new notion of privacy called \textit{aggregated privacy} which allows certain \emph{aggregation} of user statistics to be revealed to the server (but not individual ones), enabling robust computations. We formulate the problem as a secure multiparty computation (MPC) problem, and construct {\bf \scheme}, an FL framework that achieves robust aggregations against a wide range of user-launched attacks, while satisfying aggregated privacy. As concrete instantiations of the proposed \scheme framework, we develop aggregated private protocols that are robust to two specific model poisoning and backdoor attacks; we provide detailed security analyses and experimental evaluations on these two protocols, demonstrating their performance and efficiency advantages over multiple baselines. 

At the core of \scheme are three key components: (1) Robust algorithms; (2) Secret sharing with Lagrange coded computing (LCC) \cite{yu2019lagrange}; and (3) Zero-knowledge proof (ZKP).
This framework is designed to incorporate robust algorithms that can compute a reliable final aggregation against potential malicious updates. Users' private data remains in perfect secrecy within a certain number of colluding users via LCC. 
In the presence of Byzantine users who may arbitrarily deviate from the protocol, \scheme further leverages SNIP \cite{corrigan2017prio}, a secret sharing-based ZKP, to ensure the consistency of local executions. Users generate proofs that are verifiable by the collaboration of all other users and the server. The server can orchestrate the users to perform the final model aggregation securely and robustly.

The contributions of this work are summarized as follows: 
\begin{itemize}
    \item We define a novel privacy notion of \textit{Aggregated Privacy} for secure aggregation, particularly in FL settings. This definition prevents any user's individual data from leakage while allowing certain aggregated statistics revealed. 
    
    \item We propose \scheme, a secure FL framework for robust aggregation that satisfies aggregated privacy, against the collusion of semi-honest server and Byzantine users. It guarantees honest users to keep local data private while receiving a reliable global model resilient to various model integrity attacks.

    \item We significantly improve the communication efficiency of \scheme, utilizing the constant-size commitment scheme \cite{jahani2023byzantine, nazirkhanova2022information} and extending the ZKP scheme to LCC.
    
    
    \item \scheme is instantiated into two secure and robust protocols.
    We provide theoretical analyses for their privacy and efficiency, and perform extensive experiments to demonstrate the protocols' advantages in terms of robustness and computational overheads.
\end{itemize}

\section{Background and Motivation}\label{Section background}

\subsection{Vanilla Federated Learning}
In a vanilla federated learning protocol, the early concept involves conducting model training locally to ensure the privacy of users' datasets. \cite{mcmahan2017communication, zhang2021survey}. 
Training is concurrently performed across all the users through an iterative process in which the users interact with the center server to jointly update the global model. A standard FL pipeline operates as follows: at $t$-th iteration,  the server shares the current state of the global model $\mathbf{w}^{(t)}$ with all the users. Upon receiving the global model, each user $i\in[N]$ trains a local model $\mathbf{w}_i^{(t)}$ on its dataset, where the corresponding local update is $\mathbf{x}_i^{(t)} \triangleq \mathbf{w}_i^{(t)} - \mathbf{w}^{(t)}$. This is typically done by employing a minibatch stochastic gradient descent rule \cite{bottou2010large} to minimize the local objective function $\mathcal{L}_i(\mathbf{w})$. Each user then uploads the local update $\mathbf{x}_i^{(t)}$ to the server.
After receiving the model updates, the server aggregates them to update the global model as follows.
\begin{equation}\label{eq: FL aggregation}
\mathbf{w}^{(t+1)}=\mathbf{w}^{(t)} - \eta^{(t)}\cdot\sum\limits_{n\in [N]} \mathbf{x}_i^{(t)},
\end{equation}
where $\eta^{(t)}$ is the learning rate. The updated global model $\mathbf{w}^{(t+1)}$ is distributed to all the users to start the next iteration. 

\subsection{Privacy-preserving Model Aggregation}\label{sec: Privacy-preserving Model Aggregation}
Contradicting the early concept of FL's privacy, studies have revealed that conventional federated learning poses risks of leaking users' data. Works such as \cite{zhu2019deep,yin2021see,wei2021gradient, lyu2020survey, so2023securing} have shown that by launching gradient inversion attacks, it is possible to reveal information about the local training data from individual users' local updates.
To address this issue, \textit{secure aggregation} is introduced to protect individual information during aggregation process and only reveal the final aggregated model ~\cite{bonawitz2017practical, bell2020secure, truex2020ldp, wei2020federated, naseri2020local, el2022differential, so2022lightsecagg, jahani2023swiftagg}. Aiming at preserving privacy, the commonly assumed threat model is honest-but-curious (semi-honest), where all parties follow the protocol but try as much as possible to infer user's individual information. Moreover, to mitigate the risk of the server manipulating the aggregation results, several subsequent secure aggregation approaches have evolved into verifiable secure aggregation methods to provide model integrity \cite{wang2023towards , buyukates2024lightverifl}.


\textit{Secret sharing} (SS) is a typical way to achieve secure aggregation \cite{brea,jahani2023byzantine,chowdhury2021eiffel,lycklama2023rofl, ma2024trusted}. It enables a user to split its local update (secret) into multiple secret shares, where the secret remains theoretically private as long as no entity holds more shares than a specified threshold. Shareholders (other users) can collectively aggregate these secret shares to compute the aggregated secret. The server collects the aggregation of shares and reconstructs the actual aggregation of local updates. 
 \textit{Differential privacy} (DP) is also popularly used in secure aggregation protocols ~\cite{bonawitz2017practical, wei2020federated, naseri2020local, el2022differential, le2023privacy}. This technique is to add controlled noise or randomness to updates before aggregation. This ensures that the aggregated results do not disclose sensitive information about local training data.
\textit{Homomorphic encryption} (HE) is another powerful cryptographic technique that allows computations to be performed on encrypted data without decrypting it, therefore, exploited in secure aggregation protocols ~\cite{fang2021privacy, chen2021homomorphic, ma2022privacy, wang2023towards}.
The server can exploit such a feature to perform a secure aggregation and decrypt the real aggregation without privacy leakage.
\textit{Zero-knowledge proof} as a common MPC technique, is used by a few verifiable aggregation protocols to provide verifiability to the model aggregations performed on the server side\cite{wang2024zkfl, wang2023towards , buyukates2024lightverifl, ma2024trusted}.


\subsection{Robust Model Aggregation}\label{sec: byzantine-resilient fl}

Aside from concerns for privacy, the issue of model robustness has been pervasively discussed in FL. 
It refers to attacks launched by malicious users who aim to sabotage the learning process from convergence or implement a backdoor to the global model, which impair the model robustness overall. 
Generally, there are two types of attacks of interest. Malicious users can launch \textit{untargeted attack}, during which they upload arbitrary false models to impede the server from completing the model aggregation correctly \cite{sun2019can, wang2020attack, lyu2020survey}. There are also omniscient untargeted attacks where the adversary is aware of other benign updates, so it can design the update to launch stronger attacks \cite{shejwalkar2021manipulating} or evade the gradient sifting algorithms \cite{cao2022mpaf}. 
The other type is \textit{targeted attack} or \textit{backdoor attack}, where the malicious users try to inject a \textit{backdoor} to the global model  without affecting the main learning task ~\cite{gu2017badnets, chen2017targeted, barni2019new, nguyen2020input, li2021invisible, zeng2021rethinking, li2022backdoor}. 

To detect such fraudulent user updates, robust FL aggregation has been introduced by an abundance of studies. At the onset of the introduction of model poisoning, a few noise-filtering-based algorithms are put forward to select honest ones across users~\cite{ghosh2020communication, karimireddy2020byzantine, velicheti2021secure, orekondy2019prediction, normthreshold}. The statistical traits of local updates are exploited to capture the malformed ones and compute a final aggregation without them. However, new attacks manipulate the updates to bypass the defense so that get selected for the final aggregation \cite{fang2020local, shejwalkar2021manipulating}. Competitively, Some recent robust aggregation protocols against poisoning attacks have been proposed to address this issue \cite{krum, guerraoui2018hidden, chen2021pois, zhao2020pdgan}.

Backdoor attacks are more stealthy and harder to detect as it requires similar update statistics to preserve the main task performance. Various backdoor defense methods have been proposed to counter deliberate manipulations on local updates, as evidenced by studies such as \cite{bagdasaryan2020backdoor, fung2018mitigating, li2021neural, zeng2021rethinking, li2021anti, ozdayi2021RLRdefending, chen2022effective, qi2023towards, yin2018byzantine, bernstein2018signsgd, ozdayi2021RLRdefending, xu2021signguard}.
In \cite{yin2018byzantine}, the coordinate-wise median and trimmed mean techniques are employed to achieve promising performance in robust aggregation. Several studies notice the impact of the sign differences of honest and malicious gradients \cite{bernstein2018signsgd, ozdayi2021RLRdefending, xu2021signguard}. Particularly, \cite{bernstein2018signsgd} is the first to utilize majority voting on gradient signs during final aggregation to enhance robustness and convergence.  Subsequent works investigate how to use the coordinate-wise sign information to guide the final aggregation; for instance, \cite{ozdayi2021RLRdefending} uses the sign to adjust the learning rate in each round to defend against backdoor attacks, while SignGuard ~\cite{xu2021signguard} clusters the signs to filter malicious gradients.

\subsection{Robust and Privacy-preserving Model Aggregation}\label{Byzantine-Resilient Secure Federated Learning}
Simultaneously addressing both aspects in \ref{sec: Privacy-preserving Model Aggregation} and \ref{sec: byzantine-resilient fl} is crucial in real-world applications; however, it presents significant challenges. Several studies have endeavored to tackle this issue~\cite{brea, naseri2020local, jahani2023byzantine, chowdhury2021eiffel, rathee2023elsa, lycklama2023rofl, ma2022shieldfl, li2024efficiently, le2023privacy, ma2024trusted, yazdinejad2024robust}. 
BREA \cite{brea} is the first privacy-preserving and robust FL protocol known, who uses secret sharing scheme to hide local updates, and multi-Krum algorithm to remove malicious updates ~\cite{krum}. \cite{jahani2023byzantine} significantly improves on the computational and communication efficiency of BREA. \cite{naseri2020local} investigates how local and central DP alleviates various attacks and offers a quantifiable framework for the trade-off between privacy and robustness. 
A privacy-preserving protocol, EIFFeL ~\cite{chowdhury2021eiffel}, provides a secret-sharing based secure aggregation solution while provide robustness by well-formedness verification on user updates. It does not require users to upload the local updates for integrity validation, instead, using a ZKP scheme for correctly performing norm integrity check locally. 
RoFL \cite{lycklama2023rofl} further extends the norm bounding validation in EIFFeL and demonstrates different norm constraints also mitigate the effect of backdoor attacks. It boosts the efficiency of the protocol in addition.  
ByITFL~\cite{xia2024byzantine} proposes a theoretically secure protocol that provides robustness by update norm validation as well. It is built upon the Byzantine-resilient protocol FLTrust \cite{cao2020fltrust} to further provide privacy by LCC, especially in scenarios involving colluding users alongside a curious server. 
Trust-ZKFL~\cite{ma2024trusted} explores a protocol based on ZKFL~\cite{wang2024zkfl} that provides user privacy, aggregation robustness and model integrity under a non-colluding threat model. It proposes a rank-based scheme and a similarity-based scheme to compute inter-user information to robustly guide the final model aggregation. \cite{yazdinejad2024robust} proposes a privacy-preserving and robust protocol by introducing a group of trusted autonomous auditing entities to perform filtering for malicious encrypted updates. It utilizes additive HE and a consensus scheme to relay the encrypted messages processed by the auditors.  Other protocols offer privacy and robustness guarantees by exploiting multiple-server architectures, such as \cite{ma2022shieldfl, rathee2023elsa, li2024efficiently, le2023privacy}.



\subsection{Challenges and Motivation}
\label{sec: challeges}

Current robust and secure protocols 
face limitations in either (a) privacy-preserving guarantee or (b) robust aggregation performance. 
(a) On one hand, the primary issue in ~\cite{brea, jahani2023byzantine, bozdemir2021privacy} is the privacy leakage of individual user information through intermediate computation results in plaintext. In BREA\cite{brea} and ByzSecAgg\cite{jahani2023byzantine}, the multi-Krum algorithm~\cite{krum} is used, which reveals the pairwise distances between the local updates. Such pairwise distances potentially provide information to a curious server to reconstruct users' local updates. In Trust-ZKFL~\cite{ma2024trusted}, the coordinate-wise distances between user updates are also leaked to the attacker under collusion. In rank-based scheme or similarity-based scheme, for two users in comparison, when one of the users colludes with the server or one user from the reconstruction subset, the corrupted party is able to retrieve another user's local model by simple computations. 
(b) On the other hand, works that provide provable privacy under user-server collusion are unable to defend against malicious user attacks. Specifically, EIFFeL \cite{chowdhury2021eiffel}, RoFL \cite{lycklama2023rofl} and \cite{yazdinejad2024robust}  perform validation functions on user updates individually; hence, their defense cannot leverage the overall statistical information across users. For example, malicious users can easily bypass detection by keeping their updates below the norm threshold while still implanting a backdoor. 


For general robust aggregations, cryptography techniques mentioned in Section \ref{sec: Privacy-preserving Model Aggregation} have limitations in different aspects.
Though HE provably addresses privacy concerns, a critical limitation is its impractical complexity \cite{fang2021privacy, chen2021homomorphic, ma2022privacy, ma2022shieldfl, xia2024byzantine, yazdinejad2024robust}. In practice, the computational cost is $10^3$ to $10^4$ times greater than that of the non-encrypted vanilla FL, and the models under consideration are restricted to simple logistic regression models in \cite{ma2022shieldfl, yazdinejad2024robust}. Hence, HE protocols suffer from extremely high complexity, rendering them highly impractical when sophisticated models are involved.
Compared with HE, DP alleviates the high complexity and provides comparable efficiency with vanilla FL, while give provably secure solutions \cite{bonawitz2017practical, wei2020federated, naseri2020local, el2022differential}. However, the major issue with DP is the random masking that hinders the detection of malicious updates. A malicious update potentially becomes indistinguishable to many robust algorithms by masquerading as a benign one using a masked update. How to evaluate and provide robustness against deliberately manipulated updates in DP still remains an open problem \cite{lyu2022privacy, el2022differential}.

To overcome these limitations, we are motivated to introduce a framework for general robust algorithms while minimizing information leakage with provable security. Specifically, we consider the strongest threat model used in single-server FL protocols, where Byzantine users can collude with the server to compromise the privacy of the victim users. Under this threat model, for a general class of robust algorithms against various attacks, we envision that the best scenario of privacy guarantee is the same as the secure aggregation, where only the final model aggregation is revealed to the server. 
To this end, we ask: \textit{For the colluding threat model, what is the minimum amount of information that can be revealed, to accomplish a robust and secure aggregation protocol?} 
To answer this question, we first characterize the minimum privacy leakage through formalizing a novel privacy notion called \textit{aggregated privacy}, which allows nothing but the aggregation of certain statistics of users' updates (including the aggregated update) to be revealed. Then, we develop \scheme, a general privacy-preserving framework for robust aggregation that satisfies aggregated privacy.
\section{Problem Formulation}\label{section:formulation}
Consider a federated learning (FL) network, with $N$ users who can directly communicate with the server. Users communicate with each other via relaying of the server, such as by the TLS protocol \cite{dierks1999tls} so that the relayed messages are kept confidential from the server \cite{li2020review,zhang2021survey}. As shown in \eqref{eq: FL aggregation}, iterating the global model can be viewed as computing a vector of dimension $d$, $\mathbf{w}\in\mathbb{R}^d$, which is flattened from the global model. A user trains a local model from the global model, and uploads local model difference $\mathbf{x}_i \in \mathbb{R}^d$, or local update, to the server. The server collects and aggregates the received updates to compute the global model for the next iteration.  

\noindent \textbf{Threat Model}. The adversary simultaneously corrupts a set of up to $T$ users and the server, and gains full access to their local states. We consider 
Byzantine users who may arbitrarily deviate from the protocol for multiple purposes; 
and also a semi-honest server who faithfully follows the protocol, but tries to infer victim users' private datasets from its received messages.
We assume that the set of malicious users is static in a single global iteration. 

We note that the honest-but-curious (semi-honest) server setting is pervasively considered in previous works on secure and robust federated learning~\cite{brea, jahani2023byzantine, bonawitz2017practical,so2022lightsecagg, liu2021privacy, ma2022shieldfl, le2023privacy, li2024efficiently, ma2024trusted, yazdinejad2024robust}. 
In comparison with works like ~\cite{shao2022dres,liu2021privacy,ma2022shieldfl,li2024efficiently, ma2024trusted, le2023privacy, yazdinejad2024robust} who consider the coexistence of semi-honest server(s) and malicious users but acting separately, we also consider the possibility of server-user collusion.  That is, an adversary can manipulate the messages sent from the corrupted users, such that the aggregated model update, observed at the curious server, reveals more information about the other victim users' private data.   
Also, for the single-server setting considered in this work, dealing with a Byzantine server who can arbitrarily deviate from the protocol is extremely challenging. Even for the protocols in \cite{chowdhury2021eiffel} and \cite{lycklama2023rofl} that claim to be secure against a malicious server actively inferring users' private data, they are vulnerable to recently developed active attacks for data leakage (see, e.g.,~\cite{pasquini2022eluding}). Hence, these protocols are essentially secure against a semi-honest server.
For a set of user updates $\{\mathbf{x}_i | i\in [N]\}$,  we denote a \textit{robust algorithm} by $\Omega$ and its output by $\mathbf{x}_{\Omega}$, i.e., $\mathbf{x}_{\Omega}\leftarrow \Omega(\{\mathbf{x}_i | i\in [N]\})$. For example, \cite{krum} optimizes $\mathbf{x}_{\Omega}$ by minimizing the distances between $\mathbf{x}_{\Omega}$ and $\mathbf{x}_i$s; and \cite{geomedian} computes $\mathbf{x}_{\Omega}$ by finding the geometric median of $\mathbf{x}_i$s. Generally for any robust model aggregation, it can be represented by defining a $\Omega$.
Then given an aggregation algorithm, the goal is to compute $\mathbf{x}_{\Omega}$ with the minimum information leakage to the adversary. 
For secure aggregation, where $\Omega$ represents update summation, protocols are designed to compute $\Sigma\mathbf{x}_i$ as the final result such that nothing else is exposed. However, 
when $\Omega$ is complicated for a robust aggregation (e.g., \cite{krum, liu2012privacy, bozdemir2021privacy, ozdayi2021RLRdefending}), additional information is often leaked to the curious server and users.  
To formalize the minimal possible amount of information leaked to the adversary, we generalize the privacy guarantee to a notion called \textit{aggregated privacy}, which requires that no additional information is exposed, except the summation of a given function of $\mathbf{x}_i$. 
Therefore, to compute $\mathbf{x}_{\Omega}$ with a sophisticated $\Omega$ in a privacy-preserving manner, we slightly relax the security requirement in secure aggregation to certain aggregated statistics of local updates, but not allowing any information leakage related to any specific user.  We formalize this notion in the following definition:

\begin{mydef}[Aggregated Privacy]
    \label{def: aggregated privacy}
    Given users' local updates $\{\mathbf{x}_i, \forall i \in [N]\}$, a robust aggregation algorithm $\Omega$, and the function $\phi_\Omega(\cdot)$ on an individual update, an aggregation protocol satisfies aggregated privacy if nothing other than the aggregation $\sum_{i\in [N]}\phi_\Omega(\mathbf{x}_i)$ is revealed to the adversary.
\end{mydef}

\noindent\textbf{Remark 1}: As a special instance, when $\Omega$ stands for a plain summation and $\phi_\Omega$ is the identity function, i.e., $\sum_{i\in [N]}\phi_\Omega(\mathbf{x}_i)=\sum_{i\in [N]}\mathbf{x}_i$, the aggregated privacy reduces to the conventional privacy constraint in secure aggregation. However, plain summation of individual updates does not provide any robustness to malicious updates.

\noindent \textbf{Remark 2}: Achieving aggregated privacy for robust aggregation algorithm $\Omega$ is highly non-trivial, and \textbf{none} of previous works have fulfilled such definition. BREA \cite{brea} and ByzSecAgg \cite{jahani2023byzantine} utilize the Krum algorithm, where the pairwise distances between $\mathbf{x}_i$ and $\mathbf{x}_j$ are revealed to the curious server, hence violating the aggregated privacy. Similarly, Trust-ZKFL \cite{ma2024trusted} reveals coordinate-wise distance or cosine similarity between user updates, to not only the server but also a subset of users. In EIFFeL \cite{chowdhury2021eiffel} and RoFL \cite{lycklama2023rofl}, a validation predicate, typically a norm bound, is applied to individual updates, revealing the validity of $\mathbf{x}_i$ to all parties. This approach violates aggregated privacy. To mitigate the privacy leakage of individual users, \cite{yazdinejad2024robust} introduces a trusted and autonomous auditing system but still reveals the validity of $\mathbf{x}_i$ to all.

To design provably secure robust protocols with aggregated privacy, we formulate a secure multi-party computation (MPC) problem with $N$ input updates $\{\mathbf{x}_i | i\in [N]\}$ from $N$ FL users, denoting the MPC problem as ${SRA}_{\Omega}$ (short for {secure and robust aggregation with $\Omega$}).
The adversary controls over a corrupted subset $U_M$ of  users and can access the local state of the server; it runs any next-message algorithm,  $\mathcal{A}$.

\begin{mydef}[${SRA}_{\Omega}$]
    \label{def: SRA}
    Given a robust aggregation algorithm $\Omega$ and inputs $\{\mathbf{x}_i | i\in [N]\}$, a protocol $\Pi_{\Omega}$ securely executes $\Omega$ satisfying aggregated privacy
    if it satisfies:
    \begin{itemize}
        \item \textbf{Integrity}. Consider the output of the protocol execution 
        $\mathbf{x}_{agg} = \Pi_{\Omega} (\{\mathbf{x}_i | i\in [N]\})$, and the plaintext execution $\mathbf{x}_{\Omega} = \Omega(\{\mathbf{x}_i | i\in [N]\})$. There always exists a proper security parameter $\kappa$, such that the following holds: 
        \begin{equation}
            Pr[|\mathbf{x}_{agg}-\mathbf{x}_{\Omega}|_2 < \tau] \geq 1-negl(\kappa),
        \end{equation}
        for any small constant $\tau>0$.
        \item \textbf{Privacy}. 
        For any adversary running $\mathcal{A}$ who controls a malicious subset $U_M$ of users and a semi-honest server, 
        and $U_H = [N]\backslash U_M$ denotes the rest of honest users,
        there exists a probabilistic polynomial-time (P.P.T.) simulator $\mathcal{S}$, such that
        \begin{multline}\label{eq: security of collusion}
            Real_{\Pi_{\Omega}}( \mathcal{A}, \{ \mathbf{x}_i | i\in U_H \}, \{ \mathbf{x}_i | i\in U_M \})
                     \equiv_{C} \\
                      \mathcal{S}({\Omega}, \{ \mathbf{x}_i | i\in U_M \}, \sum_{i\in [N]}\phi_{\Omega} (\mathbf{x}_i) ),
        \end{multline}
        where $Real_{\Pi_{\Omega}}$ represents the joint view of $U_M\cup \{server\}$ executing $\Pi_{\Omega}$; $\phi_{\Omega}$ is a pre-defined function applied to updates; ``$\equiv_{C}$'' denotes computational indistinguishability.
  \end{itemize} 
\end{mydef}

\noindent {\bf The Proposed Solution.} We propose a secure computing framework \textit{\scheme} to construct ${SRA}_{\Omega}$ protocols. The key idea of \scheme is to choose an appropriate function $\phi$, such that the robust aggregation can be accomplished via computing aggregation of $\phi({\bf x}_i)$s. To ensure privacy and integrity for the aggregation operation against malicious adversaries,
\scheme mainly relies on the following techniques: 
\begin{enumerate}
    \item \emph{Lagrange coded computing}~\cite{yu2019lagrange} enables efficient message encoding with perfect secrecy, while supporting linear computations on the ciphertext.  We exploit this technique to create secret shares of the updates and their functions, as well as the aggregation of the updates and functions, such that they are kept private during the aggregation process. 
    \item \emph{Secret-shared non-interactive proof}~\cite{corrigan2017prio} provides a distributed ZKP for the authenticity of the users' operations. We use this technique to guarantee that the users have honestly performed the function evaluations on their shares. 
\end{enumerate}

Next, we first present cryptographic primitives employed by \scheme in Section~\ref{sec: building block}, and then provide the full description of \scheme and its instantiations for two particular $\Omega$s in Section~\ref{sec: framework}.

\section{Building Blocks of \scheme}\label{sec: building block}

\subsection{Key Agreement (\texttt{KA}) and Authenticated Encryption (\texttt{AE})} 
We use the standard Key agreement (\texttt{KA}) and Authenticated encryption (\texttt{AE}) for user communication's privacy and integrity. \texttt{KA} and \texttt{AE} involve three and two algorithms respectively.
\begin{itemize}
    \item $pp \leftarrow$ \texttt{KA.param($\kappa$)}. It generates public parameters $pp$ from security parameter $\kappa$. 
    \item $[pk_i,sk_i]\leftarrow$ \texttt{KA.gen($pp$)}. The key generation algorithm creates a public-secret key pair.
    \item $sk_{i,j}\leftarrow$ \texttt{KA.agree($pk_j,sk_i$)}. User $i$ combines two keys to generate a symmetric session key $sk_{i,j}$. 
    \item  $c\leftarrow$\texttt{AE.enc($k,m$)}. The encryption algorithm uses a message $m$ with key $k$ and returns the ciphertext $c$. 
    \item $m\leftarrow$\texttt{AE.dec($k,c$)}. The decryption algorithm takes a ciphertext $c$ with key $k$ and returns the original text $m$, or error symbol $\perp$ when decryption fails. 
\end{itemize}


\subsection{Secret Sharing with Lagrange Encoding} \label{LCC_ss} The proposed protocol relies on a secret sharing scheme with encoding, which is known as Lagrange coded computation (LCC)~\cite{yu2019lagrange}. We extend it to the verifiable LCC with constant-size commitment in the next subsection. 

Initially, the server and all users agree on $K+T+N$ distinct elements $\{\alpha_{i},\beta_{k}:i\in[N],k\in[K+T]\}$ from $\mathbb{F}_q$ in advance.
User $i$ holds a vector $\mathbf{x}_i\in \mathbb{F}_q^d$, where $d$ is the dimension of the vector. Then, user $i$ partitions its message ${\mathbf{x}}_i$ into $K$ equal-size pieces, labeled by ${\mathbf{x}}_{i,1},{\mathbf{x}}_{i,2},\ldots,{\mathbf{x}}_{i,K}\in\mathbb{F}_q^{\frac{d}{K}}$. To provide $T$-colluding privacy protection, the user $i$ independently and uniformly generates $T$ random noises $\mathbf{z}_{i,K+1},\ldots,\mathbf{z}_{i,K+T}$ from $\mathbb{F}_q^{\frac{d}{K}}$. These random noises are used for masking the partitioned pieces, which are achieved by creating a polynomial $f_{i}:\mathbb{F}_{q}^{\frac{d}{K}}\rightarrow\mathbb{F}_q^{\frac{d}{K}}$ of degree $K+T-1$ such that\footnote{Unless otherwise specified, all operations are taken modulo $q$.}

\begin{IEEEeqnarray}{l}
f_i(\beta_{i})=\left\{
\begin{array}{@{}ll}
{\mathbf{x}}_{i,k},&\forall\, k\in[K]\\
\mathbf{z}_{i,k},&\forall\, k\in[K+1:K+T]
\end{array}\right.. \label{encdoing goal}
\end{IEEEeqnarray}
The Lagrange interpolation rules and the degree restriction guarantee the existence and uniqueness of the polynomial $f_i(\alpha)$, which has the form of

\begin{IEEEeqnarray}{c}\label{encoding polynomial}
\begin{split}
f_i(\alpha)&=\sum\limits_{k=1}^{K}{\mathbf{x}}_{i,k}\prod_{m\in[K+T]\backslash\{k\}}\frac{\alpha-\beta_{m}}{\beta_{k}-\beta_{m}}\\
&+\sum\limits_{k=K+1}^{K+T}\mathbf{z}_{i,k}\prod_{m\in[K+T]\backslash\{k\}}\frac{\alpha-\beta_{m}}{\beta_{k}-\beta_{m}}.
\end{split}
\end{IEEEeqnarray}

Next, user $i$ evaluates $f_{i}(\alpha)$ at point $\alpha=\alpha_{j}$ to obtain an encoded piece for user $j$, by
\begin{IEEEeqnarray}{c}\label{sharing piece}
[\mathbf{x}_i]_{j} = f_{i}(\alpha_{j}),\quad\forall\,i,j\in[N].
\end{IEEEeqnarray}
We use $[\cdot]$ to distinguish the original data and its encoded shares \footnote{When user indices appear on both sides of $[\cdots_i]_j$, it indicates the user $i$'s secret share sent to user $j$. We abuse $[\cdot]$ to present vectors without such indices.}. Taking $[\mathbf{x}_i]_{j}$ in \eqref{sharing piece} for an example, $\mathbf{x}_i$ is the user $i$'s secret, and $j$ stands for the share holder's index, therefore, $[\mathbf{x}_i]_{j}$ is what user $j$ holds for secret $\mathbf{x}_i$. Here, we allow $i=j$ for protocol simplicity. 
Finally, we wrap \texttt{LCC} as a primitive module that includes two algorithms:
\begin{itemize}
    \item $[ [\mathbf{x}_i]_{1},\ldots,[\mathbf{x}_i]_{N}] \leftarrow$ \texttt{LCC.share($\mathbf{x}_i$)}. Given the threshold $T$ and parameter $K$, where $T+K<N$, it generates shares for a secret vector $\mathbf{x}_i$. The algorithm partitions a vector into $K$ pieces,  uses \eqref{encoding polynomial} to encode them, and then generates shares using \eqref{sharing piece}.
    \item $ \mathbf{x}_i \leftarrow$ \texttt{LCC.recon($[[\mathbf{x}_i]_{1},\ldots,[\mathbf{x}_i]_{N}]$)}. The LCC generates a $[N,K+T,N-K-T+1]$ Reed-Solomon code. It can be decoded by the Gao's decoding algorithm~\cite{gao2003new}. 
\end{itemize}

\subsection{Verifiable LCC with Constant-Size Commitments}

We combine LCC and constant-size commitments ~\cite{jahani2023byzantine, nazirkhanova2022information} to leverage threshold privacy and encoding efficiency, which further reduces the communication overhead.

Let $g$ be a generator of the cyclic group $\mathbb{G}$ of order $\lambda$, where $\lambda$ should be some enough large prime such that $q$ divides $\lambda-1$.
To verify that the piece $\mathbf{x}_{i, k}$ is correctly constructed by using the Lagrange polynomial from \eqref{encoding polynomial}, user $i$ broadcasts the commitments $[{c}_{i,1},\ldots,{c}_{i,K+T}]$ to all other parties, by
\begin{IEEEeqnarray}{c}\label{commitments}
c_{i,k}= \begin{cases}\prod_{p=1}^{\left \lceil  \frac{d}{K} \right \rceil}\left(g^{\gamma^{p-1}}\right)^{\mathbf{x}_{i, k}[p]}, & \text { if } k \in[K], \\ \prod_{p=1}^{\left \lceil  \frac{d}{K} \right \rceil}\left(g^{\gamma^{p-1}}\right)^{\mathbf{z}_{i, k}[p]}, & \text { if } k \in[K+1: K+T],\end{cases}
\end{IEEEeqnarray}
where $\mathbf{x}_{i, k}[p]$ is the $p$-th entry of $\mathbf{x}_{i, k}$. A trusted third party generates $\gamma\in\mathbf{F}_q$ and computes $\mathbf{B}  \triangleq \left[g^{\gamma^0}, g^{\gamma^1}, \ldots, g^{\gamma^{\frac{d}{K}-1}}\right]$, where $\mathbf{B}$ is public to all users.

Upon receiving the share \eqref{sharing piece} and the public commitments \eqref{commitments} from user $i$, the user $j$  can verify the share by checking
\begin{equation}\label{verification}
\prod_{p=1}^{\left \lceil  \frac{d}{K} \right \rceil}\left(g^{\gamma^{p-1}}\right)^{[\mathbf{x}_i]_{j}[p]} \stackrel{?}{=} \prod_{k=1}^{K+T} {c}_{i, k} ^{\left(\prod_{m \in[K+T] \backslash\{k\}} \frac{\alpha_{j}-\beta_m}{\beta_k-\beta_m}\right)}.
\end{equation}
The equality ensures that share $[\mathbf{x}_i]_{j}$ is generated correctly by the polynomial in \eqref{encoding polynomial}. Moreover, assuming the intractability of computing the discrete logarithm for all users and the server, i.e., they cannot compute the discrete logarithm $\log_{g}({c}_{i,k})$ for any $k\in[K]$, therefore, there is no information leakage about the user's message. Following \eqref{commitments} and \eqref{verification}, we define two more algorithms for verifiability of \texttt{LCC}: 
\begin{itemize}
    \item $ [c_{i,1},\ldots,c_{i,K+T}] \leftarrow$ \texttt{LCC.commit($\mathbf{x}_i$)}. Given $\mathbf{x}_i$ and LCC encoding as described in \eqref{encoding polynomial}, the algorithm generates a constant-size vector commitment with dimension $K+T$, denoted as $\mathbf{c}_{\mathbf{x}_i}=[c_{i,1},\ldots,c_{i,K+T}]$ as in \eqref{commitments}. 
    \item $\{\perp,1\} \leftarrow$ \texttt{LCC.verify($[\mathbf{x}_{j}]_i,\mathbf{c}_{\mathbf{x}_j}$)}. This algorithm allows user $i$ to verify whether the received share is consistent with the commitments opened by user $j$, as defined in \eqref{verification}. The function returns a Boolean value of $1$ to indicate success in verification, or the error symbol $\perp$ if verification fails.
\end{itemize}



\subsection{Secret-shared Non-interactive Proof (SNIP)} \label{sec: SNIP}

We introduce the \textit{secret-shared non-interactive proof} (SNIP) as a primitive \cite{corrigan2017prio}. In the FL setting, SNIP can be used to prove that an arithmetic circuit is evaluated on certain local data that is secret shared. SNIP is originally proposed in additive secret sharing, EIFFeL extends it to the Shamir secret sharing \cite{chowdhury2021eiffel} and we further extend it to LCC for \scheme. 

\texttt{SNIP} follows the \texttt{LCC} module's setting. Given an arithmetic circuit agreed by all parties, denoted by $f(\cdot)$, the goal is for a user (prover) to prove it has faithfully performed $f(\mathbf{x}_i)=\mathbf{e}_i$ without revealing its inputs $\mathbf{x}_i$ and outputs $\mathbf{e}_i$,
Without loss of generality, we take a user to act as the prover while all the other users act as verifiers. 
Assume that the arithmetic circuit of $f$ has $P$ multiplication gates, where each gate has two input wires and one output wire. The topological order of the multiplication gates is given by $[1,2,\ldots,P]$ and agreed upon by all users. The \texttt{SNIP} works in the following steps.
\begin{enumerate}
    \item A prover (user $i$) evaluates function $f$, $f(\mathbf{x}_i) = \mathbf{e}_i$.  For the circuit multiplication gates, denote two input wires for gates $1$ to $P$ as $\{u_{1}^i,\ldots,u_{P}^i\}$ and $\{v_{1}^i,\ldots,v_{P}^i\}$. Use each side of input wires to interpolate polynomials, resulting in two polynomials of degree-$(P-1)$, $f_{u}^i(t)$ and $f_{v}^i(t)$, such that  $f_{u}^i(p)=u_{p}^i,f_{v}^i(p)=v_{p}^i, \forall p\in[P]$. The product of two polynomials $f_{h}^i(t) \triangleq f_{u}^i(t)\cdot f_{v}^i(t)$ is exactly the output wire of the $p$-th gate at the evaluation $t=p$. Denote the coefficients of polynomial $f_{h}^i(t)$ by $\mathbf{h}_i\triangleq[h_{1}^i,\ldots,h_{2P-2}^i]$. The prover then generates secret shares for $\mathbf{h}_i$ by \texttt{LCC.share($\mathbf{h}_i$)} with $K=1$, which is equivalent to Shamir's secret sharing, where shares are denoted by $[\mathbf{h}_i]_j, \forall j\in[N]$.
    
    \item A verifier (user $j$) recovers the computation of the function by the shares and generates a share for the Schwartz-Zippel polynomial test ~\cite{schwartz1980fast, zippel1979probabilistic}. Specifically, with $[\mathbf{h}_{i}]_j $, the verifier interpolates a polynomial up to the degree of $2P-2$, $[f_{h}^{i}]_j(t)$, which corresponds to all output wires. Thus, the share of the circuit output, $[\mathbf{e}_i]_j$, can be evaluated on the polynomial. Next, with $[\mathbf{x}_{i}]_j$ and $[f_{h}^{i}]_j(p), p\in[P]$, it can recover all intermediate results for input wires of all gates as additions and scalar multiplications can be directly applied. Since they are all computed from the share of polynomials, we denote the input wire values by $[u_{p}^{i}]_j, [v_{p}^{i}]_j$. SNIP interpolates two polynomials for two ways of input wires from $1$ to $P$, $[f_{u}^{i}]_j(t)$, $[f_{v}^{i}]_j(t)$.  Beaver's protocol \cite{beaver1992efficient}  is used to compute $[f_{u}^{i}(t)*f_{v}^{i}(t)]_j$ from $[f_{u}^{i}]_j(t)$ and $[f_{v}^{i}]_j(t)$. All verifiers agree on one or multiple random $r$ for the Schwartz-Zippel polynomial test. Recall that the output wire's polynomial is given by $[f_{h}^{i}]_j(t)$.
    Given an $r$, a share of polynomial test is generated by computing $[\sigma_{i}]_j = [f_{u}^{i}(r)*f_{v}^{i}(r)]_j - [f_h^{i}(r)]_j$.
    
    \item The server collects verifiers' $[\sigma_{i}]_j$ to reconstruct the result of the polynomial identity test. For a prover (user $i$), the server decodes $\sigma_i$ from $\{ [\sigma_i]_{j}|{j}\in[N]$ by \texttt{LCC.recon}. If the user honestly performs the circuit, the reconstruction is $0$. Otherwise, user $i$ fails in the proof. 
    
\end{enumerate}

We summarize the described scheme as a module \texttt{SNIP} with two functions for prover and verifier correspondingly. 
\begin{itemize}
    \item $\mathbf{h}_i\leftarrow$ \texttt{SNIP.prove}$\left(f(\mathbf{x}_i) = \mathbf{e}_i\right)$. $f$ is the given arithmetic circuit a prover (user $i$) wants to prove with input $\mathbf{x}_i$ and output $\mathbf{e}_i$. \footnote{The Beaver's triplet is generated and broadcast by each user insetup phase. }
    Algorithm output $\mathbf{h}_i$ is the coefficients given at the end of step 1.

    \item $[\sigma_j]_i, [\mathbf{e}_j]_i \leftarrow$ \texttt{SNIP.recon}($f,[\mathbf{x}_{j}]_i,[\mathbf{h}_{j}]_i$). This algorithm constructs the share of the polynomial test and the share of the circuit output as in step 2. 
    
\end{itemize}

\section{Detailed Description of \scheme}\label{sec: framework}

We first describe the general \scheme framework and then provide two concrete instantiations for robust aggregation against backdoor attacks and model poisoning attacks, respectively. 
For a given robust algorithm $\Omega$, to compute aggregation $\mathbf{x}_\Omega$ in a privacy-preserving and fault-tolerant manner, \scheme exploits LCC to secret share individual updates across users and chooses an appropriate $\phi_\Omega$ such that $\mathbf{x}_\Omega$ can be computed from $\sum_{i\in[N]}\phi_\Omega(\mathbf{x}_i)$. To guarantee the correctness in the presence of Byzantine users, \scheme utilizes SNIP to verify the computations of the arithmetic circuit $f_\phi$ induced from function $\phi$. 
\scheme proceeds in three main rounds: in the first round, the users perform local training to acquire updates, then they secret share updates and SNIP proofs; in the second round, the users perform local computations on the acquired shares and send results to the server; in the last round, the server processes the received results and compute the final robust aggregation. 
A full pipeline of \scheme is presented in Fig.\ref{fig:scheme}


\subsection{General \scheme Framework}
\scheme starts with a \textbf{setup phase} to prepare all parties with some public parameters, which includes the security parameter $\kappa$ for key initializations, the finite field $\mathbb{F}_q$ with sufficiently large $q$ and amplifier $p\in\mathbb{F}_q$ for update computations, threshold $T$ for the maximum number of malicious users, partitioning parameter $K$ for LCC, vector $\mathbf{B}$ for constant-size commitment scheme, a set of points $\{\alpha_{i},\beta_k\in\mathbb{F}_q | i\in[N],k\in[K+T]\}$ for LCC evaluations.  

\noindent \textbf{Round 1 (Users Training and Secret Sharing)}.  After acquiring the update from local model training, user $i$ first maps the float local update vector $\mathbf{x}_i^{float}$ to a finite-field vector $\mathbf{x}_i \in \mathbb{F}_q^d$ following a standard \texttt{quantize} technique as in ~\cite{brea, li2023fedvs} with amplifier $p$.
User $i$ uses \texttt{LCC.share} to slice and encode each $\mathbf{x}_i$ into $N$ shares $\left\{\left[\mathbf{x}_i\right]_1,\left[\mathbf{x}_i\right]_2, \cdots\left[\mathbf{x}_i\right]_N\right\}$. To ensure the authenticity of the shared information, user $i$ commits to its shares through the constant-size commitment scheme and generates commitment $\mathbf{c}_{x,i}$.  
Next, a user needs to prove that it has faithfully evaluated an arithmetic circuit $f_{\phi}(\mathbf{x}_i)=\mathbf{e}_i$, where $f_{\phi}$ is defined based on $\phi_\Omega$. 
Using \texttt{SNIP.prove}, the user $i$ generates its proof $\mathbf{h}_i$ for the equality $f_{\phi}(\mathbf{x}_i)=\mathbf{e}_i$, then encodes $\mathbf{h}_i$ to $\{[\mathbf{h}_i]_j,\forall j\in[N]\}$ and commits to it by $\mathbf{c}_{{h},i}$. 
At the end of the round, each user $i$ sends \texttt{AE.enc}($[\mathbf{x}_i]_j, [\mathbf{h}_i]_j$) to user $j,\forall j\in[N]$ encrypted by $(i,j)$ secret session key, and broadcasts commitments $\mathbf{c}_{x,i}, \mathbf{c}_{{h},i}$.


\noindent\textbf{Round 2 (Users Local Computation)}. Using the broadcast commitments, user $i$ can verify the received shares from any user $j$. If user $j$'s message does not pass the verification, user $i$ adds user $j$ to the local malicious user set $U_A^{(i)}$. Next, user $i$ uses \texttt{SNIP.recon} to construct a share of the polynomial identity check for user $j$, denoted by $[\sigma_j]_i$.
User $i$ combines the update shares over all users $\sum_{j\in[N]}[\mathbf{x}_j]_i$ as the aggregation of the shares. Then, the user $i$ submits the following four messages to the server: $m_1$: the local malicious user set $U_A^{(i)}$;
$m_2$: the shares of polynomial identity check $\{ [\sigma_j]_i, \forall j \in [N] \}$;
$m_3$: summation of $f_\phi$ output's shares $\sum_{j\in[N]} [\mathbf{e}_j]_i$; and $m_4$: the aggregation of updates' shares $\sum_{j\in[N]} [\mathbf{x}_j]_i$. 


\noindent\textbf{Round 3 (Server's Computation and Aggregation)}. The server collects users' messages and performs consistency checks to construct a global malicious user set $U_A$. Specifically, for message $m_1$, the server looks for user $i$ who satisfies any of two cases to add $i$ into $U_A$. The first case is when a user is tagged more than $T$ times, then it must have been tagged by at least one honest user. The second case is when a user tags more than $T$ other users, then it must have tagged at least one honest user as malicious. 
For message $m_2$, the server decodes each $\sigma_i, \forall i \in [N]$ and checks the reconstructed value. If $\sigma_i$ is zero, it means that the circuit is correctly executed by user $i$. If it is not, the circuit computations fail in at least one multiplication gate for user $i$, so $i$ will be added to $U_A$.  With the message $m_3$, the server reconstructs the aggregation of outputs $\sum_{j\in[N]} \mathbf{e}_j$. It is used as the guidance for robust aggregation of updates. How the aggregated information guides the final robust aggregation is expanded in the next section.
After the server computes $\mathbf{x}_\Omega$, it uses  \texttt{dequantize($\mathbf{x}_\Omega$)}\footnote{\texttt{dequantize} is the inverse function of \texttt{quantize} that maps values in $\mathbb{F}_q^d$ to $\mathbb{R}^d$.} to map the aggregation back to the real domain for the global model updating as in \eqref{eq: FL aggregation}. Before releasing the final aggregation, the server checks the global malicious user set $U_A$. It concludes the current iteration if $|U_A|=0$, or aborts the current iteration if $|U_A|> T$. However, if $0<|U_A|\le T$, the server will invoke an additional execution of round $2$ by sending its global malicious set $U_A$ to all users, so that users should recalculate each message with respect to set $[N]\backslash U_A$ instead of $[N]$. Then, the server normally proceeds round $3$ with the malicious users removed. 

\subsection{Two Instantiations}
Leveraging the proposed framework \scheme, we construct two concrete $SRA$ protocols for two specific robust aggregation algorithms. One is Robust Learning Rate (RLR) for backdoor attack, and the other is Robust Federated Aggregation (RFA) for model poisoning attack. 
For each robust aggregation algorithm, we first explain how the robust algorithms operate in plaintext, and then present how the robust algorithms apply to \scheme to achieve the $SRA$ protocols.
Additionally, a special feature of repeated pattern circuits is used in both instantiations to reduce SNIP's computational overhead. Detail is provided in Appendix \ref{app:snip}.

\noindent {\bf \scheme with RLR.} In RLR~\cite{ozdayi2021RLRdefending}, it is shown to be effective to use the sign information of updates to avoid the backdoor behaviors.
For each coordinate of final aggregation, the positive/negative sign is decided on the number of positive/negative signs of all user updates. The users' majority vote of signs generates a binary mask with values in  $\{\pm1\}$. The final aggregation's coordinate signs are determined by applying this binary mask to eliminate the malicious updates.  Specifically, for $\mathbf{x}_i$ and the $k$-th entry $\mathbf{x}_i[k]$, RLR assigns $+1$ to a binary vector $\mathbf{v}$ if  $\max\{\sum_{i\in[N]} sign(\mathbf{x}_i[k]), N-\sum_{i\in[N]} sign(\mathbf{x}_i[k]) \}\ge\theta$
, or assigns $-1$ otherwise, where the threshold is $\theta = T+1$ and $sign(\cdot)$ returns $1$ or $0$ for positive or negative inputs, respectively. The final aggregation result is $\mathbf{x}_{RLR} = \mathbf{v} \odot \sum_{i\in[N]}\mathbf{x}_i$, where $\odot$ is the element-wise multiplication. 

Applying \scheme with RLR, we let $\phi$ be: 
\begin{equation}\label{eq: phi rlr}
    \phi(\mathbf{x}_i) = \Big(sign(\mathbf{x}_i) 
    , \mathbf{x}_i \Big) .
\end{equation}
In this case, the evaluation circuit to prove is simply $f_\phi(\mathbf{x}_i)\triangleq sign(\mathbf{x}_i)=\mathbf{e}_i$.
As a result, the server reconstructs $\sum_{i\in[N]} \mathbf{e}_i$ and $\sum_{i\in[N]} \mathbf{x}_i$ in the last round, with no additional information revealed. Next, it recovers the binary mask $\mathbf{v}$ from $\sum_{i\in[N]} \mathbf{e}_i$ to correctly compute final robust aggregation as in RLR.

\begin{figure*}[]
    \centering
    \includegraphics[width=0.85\linewidth]{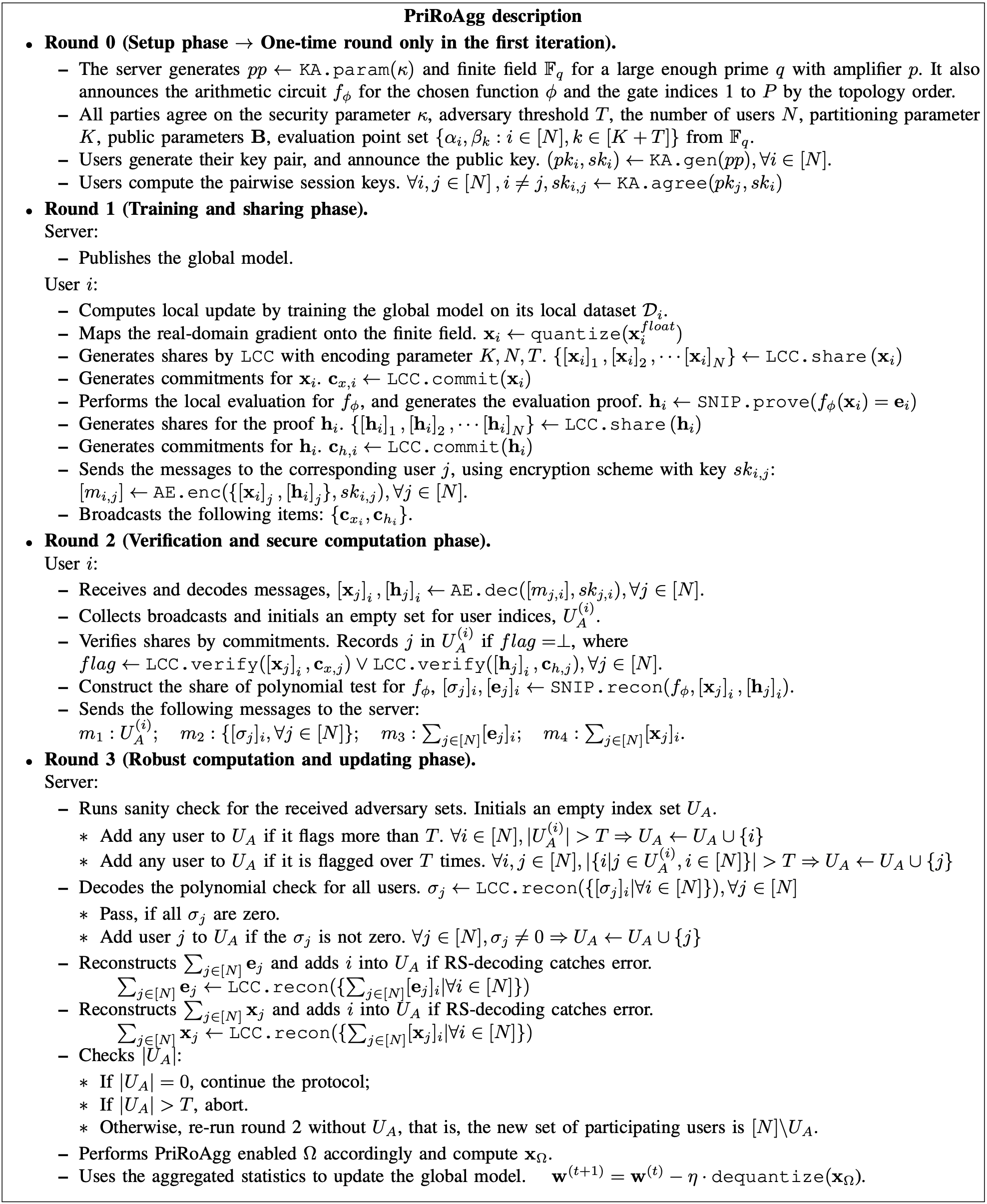}
    \caption{Description of framework \scheme}
    \vspace{-5mm}
    \label{fig:scheme}
\end{figure*}

\noindent {\bf \scheme with RFA.} RFA~\cite{geomedian} exploits the Weiszfeld algorithms \cite{weiszfeld2009point} to compute the geometric median for the user updates. Specifically, we utilize the one-step smoothed Weiszfeld algorithm (Algorithm~3 in \cite{geomedian}) to compute the approximation of the geometric median. RFA computes $\omega_i = 1/||\mathbf{x}_i||_2$ and uses $\omega_i$ as the weight for the final robust aggregation. More precisely, $\mathbf{x}_{RFA} = \frac{1}{\sum_{i\in[N]} \omega_i} \sum\limits_{i\in [N]}\omega_i \mathbf{x}_i$.

In RFA instantiation, we choose $\phi$ to be :
\begin{equation}\label{eq: phi rfa}
    \phi(\mathbf{x}_i) = \Big(1/||\mathbf{x}_i||_2 , \omega_i\mathbf{x}_i \Big) ,
\end{equation}
where $||\mathbf{x}_i||_2 \neq 0$. 
Since the arithmetic circuit does not allow division, the circuit to be verified is designed as $f_\phi({\mathbf{x}}_i)\triangleq(1 - \omega_i||\mathbf{x}_i||_2)^2=\mathbf{e}_i$, where $\mathbf{e}_i = 0$. 
In round $2$, user $i$ computes $[\omega_j\mathbf{x}_j]_i$ from $[\omega_j]_i$ and $[\mathbf{x}_j]_i$ by the Beaver's protocol~\cite{beaver1992efficient}. At the end of round 2, user $i$ substitutes $\sum_{j\in[N]} [\mathbf{x}_j]_i$ by concatenating $\sum_{j\in[N]} [\omega_j\mathbf{x}_j]_i$ and $\sum_{j\in[N]} [\omega_j]_i$ as message $m_4$. 
In the final round, the server reconstructs and checks if $\sum_{i\in[N]} \mathbf{e}_i=0$. Together with the fact check of polynomial identity $\sigma_i=0$, the server ensures the circuit $f_\phi(\mathbf{x}_i)$ is faithfully evaluated. For the final aggregation, the server reconstructs and decomposes $m_4$ to compute the weighted average $\frac{1}{\sum_{i\in[N]} \omega_i} \sum\limits_{i\in [N]}\omega_i \mathbf{x}_i$ as in RFA.





\section{Theoretical Analysis}\label{sec: theory}
\subsection{Security analysis}


In this subsection, we provide the security proofs for the privacy guarantees of two protocols in \scheme. 



\begin{mythm}(${SRA}_{RLR}$)
    \label{thm: SRA RLR}
    For a protocol $\Pi_{RLR}$ that securely implements a ${SRA}_{\Omega}$, assuming an adversary $\mathcal{A}$ who controls a subset $U_M$ of malicious users and gains full access to the local state of the server, and denoting $U_H = [N]\backslash U_M$ as the subset of honest users, there exists a probabilistic polynomial-time (P.P.T.) simulator $\mathcal{S}$, such that
    \begin{multline}\label{eq: RLR security}
            Real_{\Pi_{RLR}}( \mathcal{A}, \{ \mathbf{x}_i | i\in U_H \}, \{ \mathbf{x}_i | i\in U_M \})
                     \equiv_{C} \\
                      \mathcal{S}({\Omega_{RLR}}, \{ \mathbf{x}_i | i\in U_M \}, \sum_{i\in [N]}\phi_{RLR} (\mathbf{x}_i) ),
    \end{multline}
    where $Real_{\Pi_{RLR}}$ represents the joint view of $U_M\cup \{server\}$ executing $\Pi_{RLR}$; $\phi_{RLR}$ is defined as \eqref{eq: phi rlr}.
    
\end{mythm}

\begin{proof}
The proof uses a hybrid argument to demonstrate the indistinguishability between the real execution and the simulator's output.
\begin{itemize}
    \item $Hyb_{0}$: This random variable represents the joint view of $U_M\cup \{server\}$ in the real world execution.
    
    \item $Hyb_1$: In this hybrid, for any honest user $i$, instead of using $sk_i$ for key generation, the simulator generates the uniformly random key in the Diffie-Hellman key exchange protocol \cite{diffie2022new} for the set-up phase. The Diffie-Hellman assumption guarantees that this hybrid is indistinguishable from the previous one. 

    \item $Hyb_2$: In this hybrid, the only difference from $Hyb_1$ is that the simulator $\mathcal{S}$ substitutes  $\mathbf{x}_i$, and $\mathbf{h}_i$ with generated $\Tilde{\mathbf{x}}_i$, and $\Tilde{\mathbf{h}}_i$ for any honest user $i$ in round 1. The simulator 
    $\mathcal{S}$ first randomly samples $\Tilde{\mathbf{e}}_i \in \{0,1\}^d$ such that
    \begin{equation*}
        \sum_{i\in U_H} \Tilde{\mathbf{e}}_i = \sum_{i\in U_H} sign(\mathbf{x}_i).
    \end{equation*}
    Then, the signs of the entries in $\Tilde{\mathbf{x}}_i$ are constrained. The simulator $\mathcal{S}$  samples $\Tilde{\mathbf{x}}_i\in \mathbb{F}_q^d $ by the uniform distribution such that
    \begin{equation*}
        \sum_{i\in U_H}  \Tilde{\mathbf{x}}_i = \sum_{i\in U_H}  \mathbf{x}_i.
    \end{equation*}

    The shares and commitments are generated accordingly to $\Tilde{\mathbf{x}}_i$. The $\Tilde{\mathbf{h}}_i$ is generated, shared and committed by protocol. Because $|U_M| \le T$, in the joint view, the adversary cannot know anything about the secret as long as the commitments are consistent with the shares. The indistinguishability is guaranteed by the property of LCC. 
    
    \item $Hyb_3$: In this hybrid, the only difference from $Hyb_2$ is that $\mathcal{S}$ substitutes the $U_H$'s messages sent to the server with the generated messages in round 2. The $U_A^{(i)}$ does not record any honest users and the rest of the users are processed by protocol. The messages $ [\sigma_j]_i, \sum_{j\in[N]} [\mathbf{e}_j]_i$, and $\sum_{j\in[N]} [\mathbf{x}_j]_i$ are computed by the protocol after substituting  $\mathbf{x}_i$,  and $\mathbf{e}_i$ with $\Tilde{\mathbf{x}}_i$, and $\Tilde{\mathbf{e}}_i$ respectively. As they are computed in the form of secret shares, the indistinguishability is guaranteed by the property of LCC.

\end{itemize}
\end{proof}


\begin{mythm}(${SRA}_{RFA}$)
    \label{thm: SRA RFA}
    For a protocol $\Pi_{RFA}$ that securely implements a ${SRA}_{\Omega}$, assuming an adversary $\mathcal{A}$ who controls a subset $U_M$ of malicious users and gains full access to the local state of the server, and denoting $U_H = [N]\backslash U_M$ as the subset of honest users, there exists a probabilistic polynomial-time (P.P.T.) simulator $\mathcal{S}$, such that
    \begin{multline}\label{eq: RFA security}
            Real_{\Pi_{RFA}}( \mathcal{A}, \{ \mathbf{x}_i | i\in U_H \}, \{ \mathbf{x}_i | i\in U_M \})
                     \equiv_{C} \\
                      \mathcal{S}({\Omega_{RFA}}, \{ \mathbf{x}_i | i\in U_M \}, \sum_{i\in [N]}\phi_{RFA} (\mathbf{x}_i) ),
    \end{multline}
    where $Real_{\Pi_{RFA}}$ represents the joint view of $U_M\cup \{server\}$ executing $\Pi_{RFA}$; $\phi_{RFA}$ is defined as \eqref{eq: phi rfa}.
\end{mythm}
\begin{proof}   
Here, we only present $Hyb_{2}$ since the other hybrids are identical to the proof of ${SRA}_{RLR}$ in Theorem \ref{thm: SRA RLR}.
\begin{itemize}
    

    \item $Hyb_2$: In this hybrid, the only difference with $Hyb_1$ is $\mathcal{S}$ substituting $\mathbf{x}_i,\omega_i,\mathbf{h}_i$ by randomly generated $\Tilde{\mathbf{x}}_i, \Tilde{\omega}_i,\Tilde{\mathbf{h}}_i$ for $U_H$ in round 1. $\mathcal{S}$ generates weight parameters $\Tilde{\omega}_i,\forall i \in U_H$, by a uniform distribution such that
    \begin{equation*}
        \sum_{i\in U_H} \Tilde{\omega}_i = \sum_{i\in U_H} \omega_i.
    \end{equation*}
    According to \eqref{eq: phi rfa}, the norm of $\Tilde{\mathbf{x}}_i$ is constrained by $||\Tilde{\mathbf{x}}_i||_2 = 1/\Tilde{\omega}_i$. $\mathcal{S}$ samples $\Tilde{\mathbf{x}}_i$ with fixed $l_2$-norm, such that
    \begin{equation*}
        \sum_{i\in U_H} \Tilde{\omega}_i \Tilde{\mathbf{x}}_i = \sum_{i\in U_H} \omega_i \mathbf{x}_i.
    \end{equation*}

    The shares and commitments are generated accordingly for the to the $\Tilde{\mathbf{x}}_i, \Tilde{\omega}_i$. The $\Tilde{\mathbf{h}}_i$ is generated, shared and committed by protocol.  Because $|U_M| \le T$, in the joint view, the adversary cannot know anything about the secret as long as the commitments are consistent with the shares. The indistinguishability is guaranteed by the property of LCC. 

    
\end{itemize}
\end{proof}

\begin{table*}[ht]
\centering
\caption{Per-iteration overhead comparison between \scheme and EIFFeL.}
\begin{tabular}{ccc}
\hline
                                         & \scheme                                           & EIFFeL                                                                     \\ \hline
User computation                         & $\mathcal{O}(\frac{N^2d}{K} + \frac{Nd}{K}\log{N})$ & $\mathcal{O} (NTd + d\log d)$                                                              \\ \hline
Server computation                       & $\mathcal{O}((N+\frac{d}{K})N\log^2{N}\log\log N)$                 & $\mathcal{O}\left((N+d)N \log ^2 N \log\log N +T d \min \left(N, T^2\right)\right)$ \\ \hline
User communication                       & $\mathcal{O}(\frac{Nd}{K})$                                   &           $\mathcal{O}(NTd)$             \\ \hline
Server communication                       & $\mathcal{O}(Nd)$                                   &           $\mathcal{O}(N^2 + Td\min(N,T^2))$                                                  \\ \hline
\end{tabular}
\vspace{-3mm}
\label{table: overhead}
\end{table*}

\subsection{Complexity Analysis}\label{sec: complexity}
The computational overhead and communication overhead per iteration are summarized in Table~\ref{table: overhead}, and the comparison with EIFFeL is presented. 
Recall that there are $N$ users with up to $T$ malicious users; the update dimension is $d$; $K$ is the partitioning parameter for LCC satisfying $K+T < N-1$; the number of circuit $f_\phi$'s multiplication gates is $\mathcal{O}(d)$.

\noindent \textbf{Computational overhead}. In round 1, the user generates secret shares for the update and commits to them by the complexity of $\mathcal{O}(N^2\frac{d}{K})$ \footnote{Pre-compute the Lagrange coefficients.}. 
The user also performs SNIP by the complexity of $\mathcal{O}(d+\frac{d}{K}\log{N})$ since the repeat-pattern circuit is used to reduce the overhead. Each circuit's proof and commitment takes $\mathcal{O}(N\log{N})$~\cite{von2003modern} with a total $\frac{d}{K}$ number of circuits. Shares and commitments of constant $\omega$ are generated in $\mathcal{O}(N(K+T))$, which is absorbed by the first term. The total complexity for each user in round 1 is $\mathcal{O}(\frac{N^2d}{K} + \frac{Nd}{K}\log{N})$.
In round 2, the user constructs shares for the polynomial identity test with complexity $\mathcal{O}(\frac{Nd}{K})$ for each prover. The secure computations are of the complexity $\mathcal{O}(\frac{d}{K})$. If each user computes the others' shares in a linear sequence, the total complexity for each user in round 2 is $\mathcal{O}(\frac{N^2d}{K})$. 
In round 3, the server decodes the proof by the complexity of $\mathcal{O}(N^2\log^2{N}\log\log N)$ and the aggregation by the complexity of $\mathcal{O}(\frac{d}{K}N\log^2{N}\log\log N)$~\cite{gao2003new}. The other operations are of the complexity $\mathcal{O}(N)$ that is absorbed by the other terms. Then, the total complexity for the server in round 3 is $\mathcal{O}((N+\frac{d}{K})N\log^2{N}\log\log N)$.

\noindent \textbf{Communication overhead}. 
In round 1, sharing the update and sharing the proof takes $\mathcal{O}(\frac{Nd}{K})$ and $\mathcal{O}(\frac{Nd}{K})$. All commitments are constant-size and, hence, absorbed by the first term. In round 2, the messages are shares of constant-size objects except the aggregations of $\phi$. $\phi$'s output dimension is $1$ and $d$ in RFA and RLR correspondingly.  Therefore, round 2 complexity is $\mathcal{O}(d)$ at maximum. Complexity of round 3 is simply  $\mathcal{O}(d)$. 

\noindent \textbf{Comparison with EIFFeL.}
In experiments, the model dimension significantly exceeds the number of users, i.e., $d>>N$, with both $T, K$ being of order $\mathcal{O}(N)$. The usage of LCC reduces the overhead in terms of $d$ to $\frac{d}{N}$. 
Therefore, \scheme's user and server computational overheads are $\mathcal{O}(Nd)$ and $\mathcal{O}(d\log^2{N}\log\log N)$, respectively, compared to EIFFeL's $\mathcal{O}(N^2d+d\log{d})$ and $\mathcal{O}\left(dN \log ^2 N \log\log N +N^2 d \right)$, respectively. For communication overhead,  the server communication remains consistent, while the user communication is substantially reduced from $\mathcal{O}(N^2d)$ to $\mathcal{O}({Nd/K})$ through the constant-size commitment scheme. 

\section{Experiments}\label{sec: exp}
We conduct experiments to demonstrate robust performance and computational overhead of \scheme instantiations. We consider several targeted and untargeted attacks on two datasets. All experiments are simulated on a single machine using Intel(R) Xeon(R) Gold 5118 CPU @ 2.30GHz and one NVIDIA GeForce RTX 4090 GPU. 
Our code is publicly available at this repository \href{https://github.com/tardisblue9/PriRoAgg/tree/main}{github.com/tardisblue9/PriRoAgg}.

\subsection{Setup}

\noindent \textbf{The untargeted and targeted attacks.} For different types of attacks, we choose various attack strategies among model update poisoning and backdoor attacks. Two trivial poisoning strategies are considered as baselines. The malicious users manipulate updates by adding Gaussian noise \cite{geomedian} or scaling by constants to their updates\cite{bhagoji2019analyzing}. Moreover, omniscient update poisoning strategies are also used, such as Min-Max and  Min-Sum attack \cite{shejwalkar2021manipulating}, where the attacker knows all benign updates to design its malformed update. Two advanced poisoning strategies are also evaluated: PoisonFL~\cite{xie2024poisonedfl} and Sine~\cite{kasyap2024sine}. 
For backdoor attacks, we consider conventional image trojan attacks \cite{liu2018trojaning} with two different pattern types and advanced backdoor attack, IBA\cite{nguyen2024iba}.
For performance comparisons, we use the global testing accuracy to evaluate the poisoning effectiveness, and then split the global test dataset into a clean test dataset and a backdoor test dataset to evaluate main task (main accuracy, MA $\uparrow$) and backdoor task accordingly (attack succuss rate, ASR $\downarrow$). For model poisoning, the main accuracy is also used for evaluation; a higher MA indicates better defense performance.
Additionally, we present the protocol's overhead with different parameter settings.


\begin{figure*}[ht]
	\centering
	\begin{minipage}{0.24\linewidth}
	     \includegraphics[width=\linewidth]{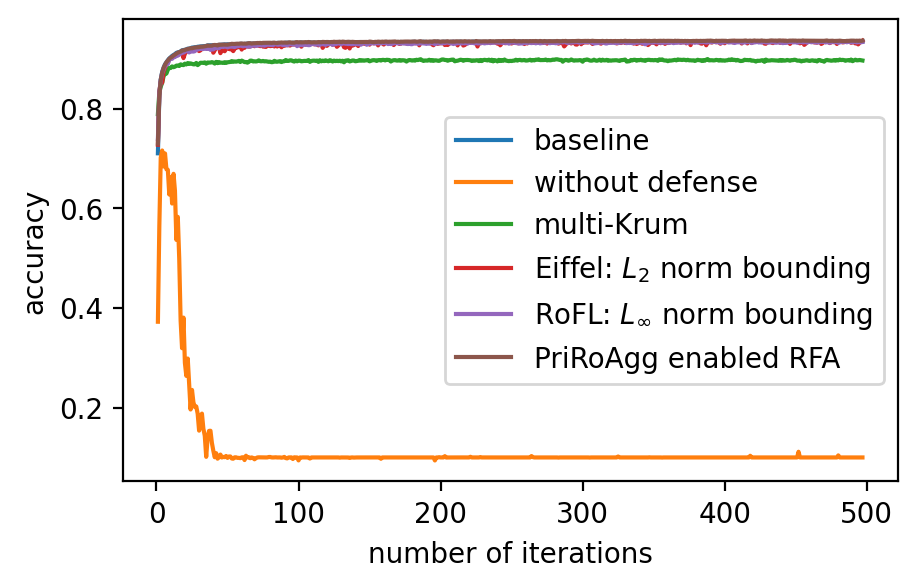}
	      \caption*{\small (a) Additive with FMNIST}
	\end{minipage}
	\begin{minipage}{0.24\linewidth}
	    \includegraphics[width=\linewidth]{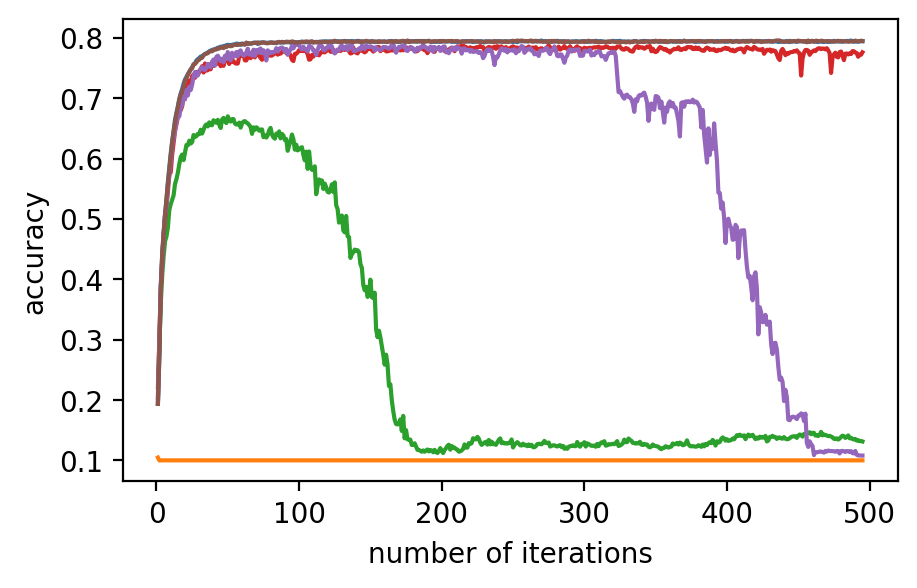}
	      \caption*{\small (b) Additive with CIFAR10}
        \end{minipage} 
	\begin{minipage}{0.24\linewidth}
	    \includegraphics[width=\linewidth]{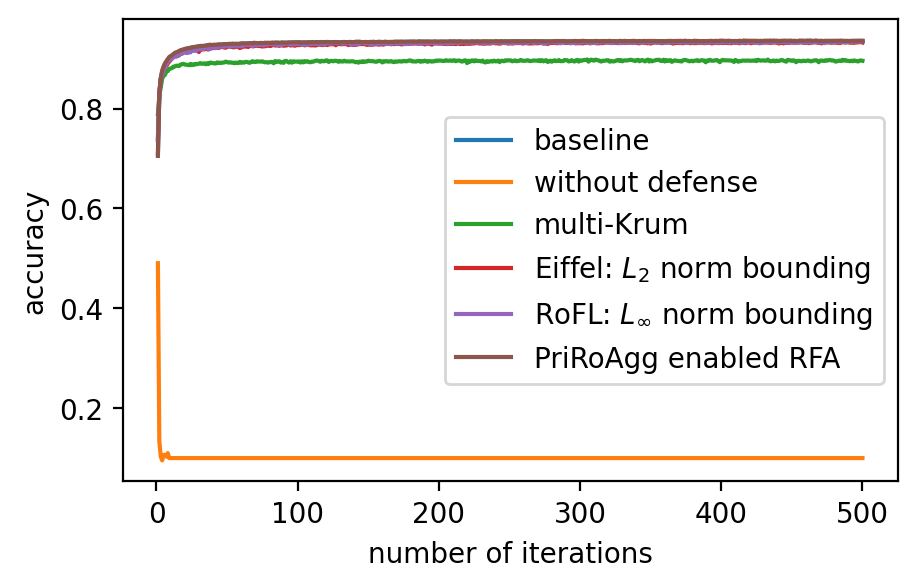}
	      \caption*{\small  (c) Scale with FMNIST }
         \end{minipage}
         \begin{minipage}{0.24\linewidth}
	    \includegraphics[width=\linewidth]{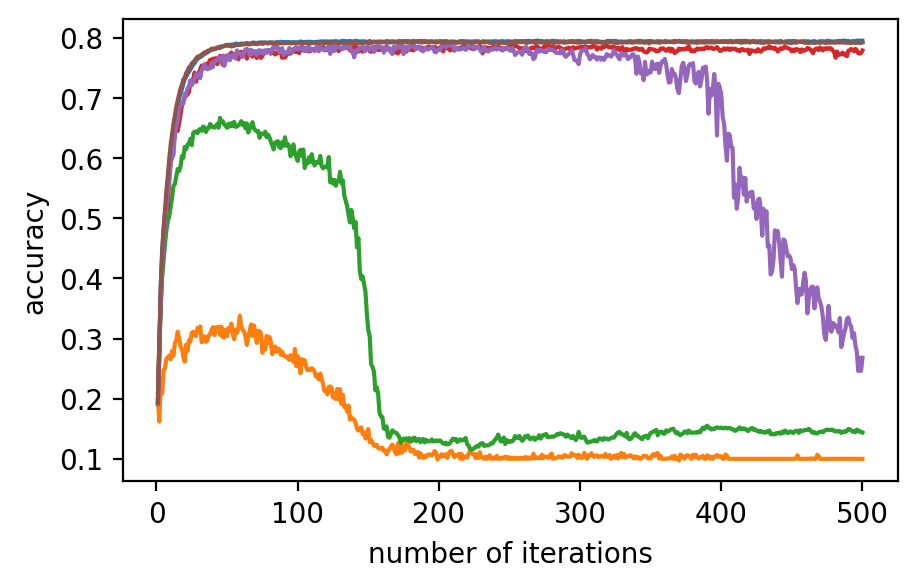}
	      \caption*{\small  (d) Scale with CIFAR10 }
         \end{minipage}\\
     \begin{minipage}{0.24\linewidth}
	     \includegraphics[width=\linewidth]{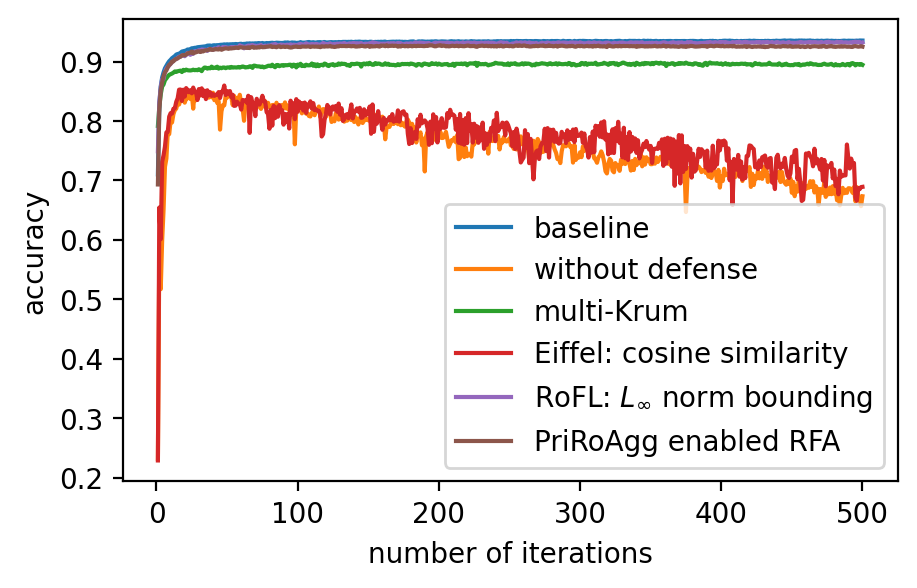}
	      \caption*{\small  (e) Min-Max with FMNIST}
	\end{minipage}
	\begin{minipage}{0.24\linewidth}
	    \includegraphics[width=\linewidth]{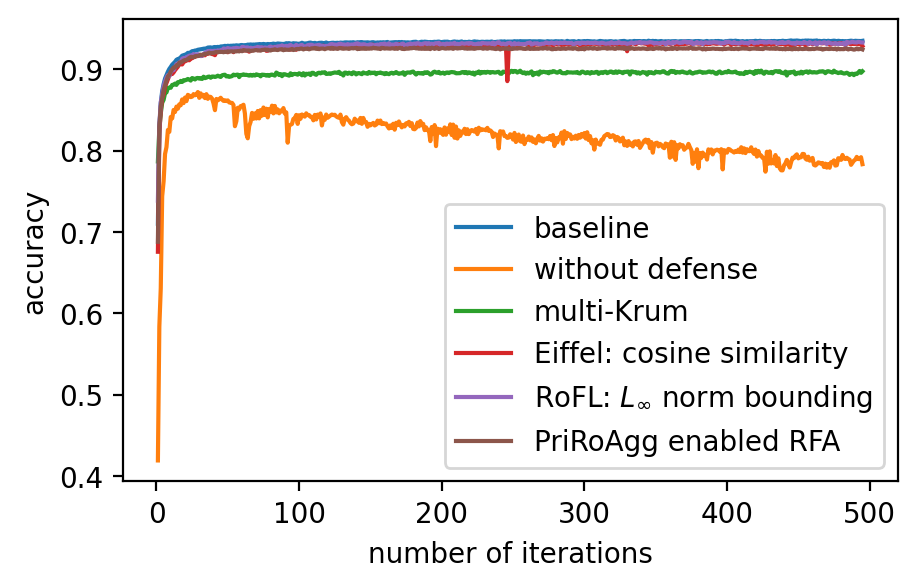}
	      \caption*{\small  (f) Min-Sum with FMNIST}
        \end{minipage} 
	\begin{minipage}{0.24\linewidth}
	    \includegraphics[width=\linewidth]{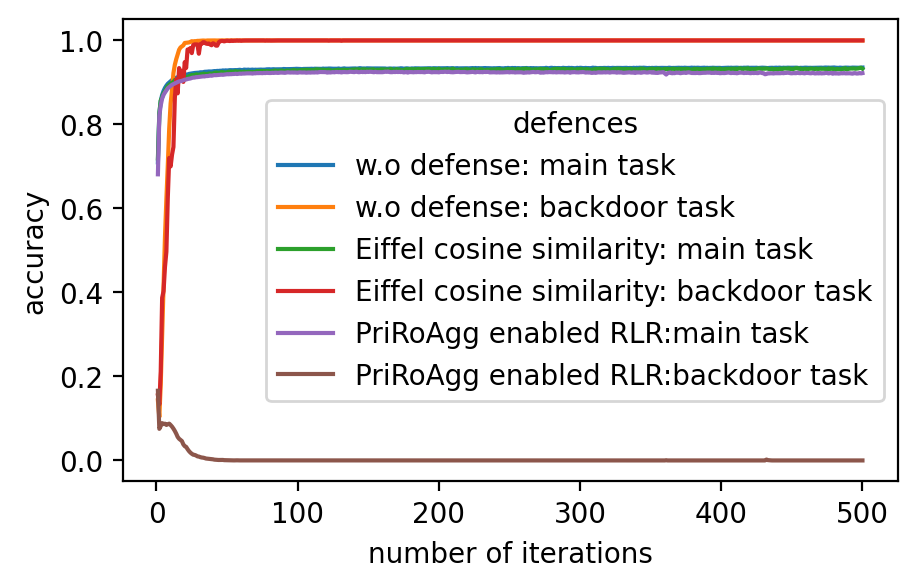}
	      \caption*{\small (g)Backdoor on FMNIST}
         \end{minipage}
         \begin{minipage}{0.24\linewidth}
	    \includegraphics[width=\linewidth]{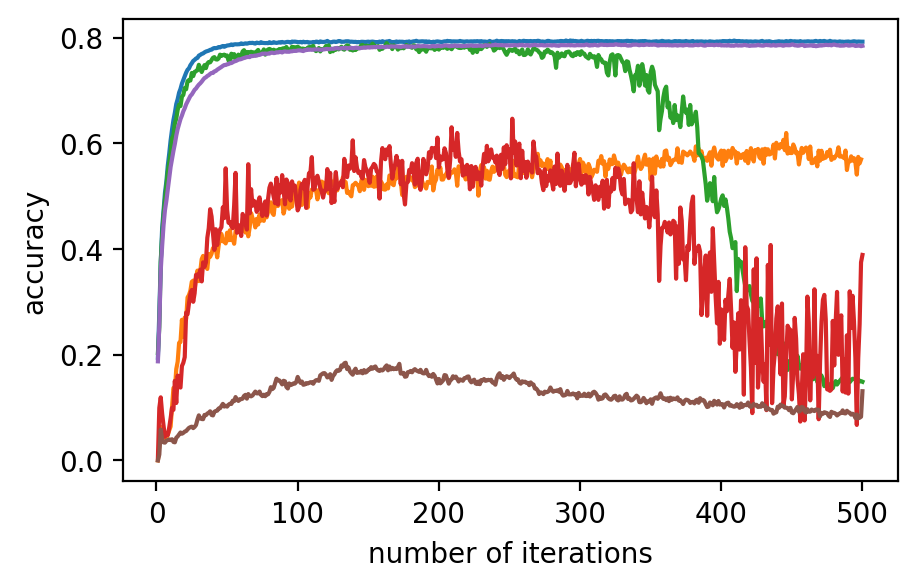}
	      \caption*{\small  (h)Backdoor on CIFAR10}
         \end{minipage}\\
    \begin{minipage}{0.24\linewidth}
	     \includegraphics[width=\linewidth]{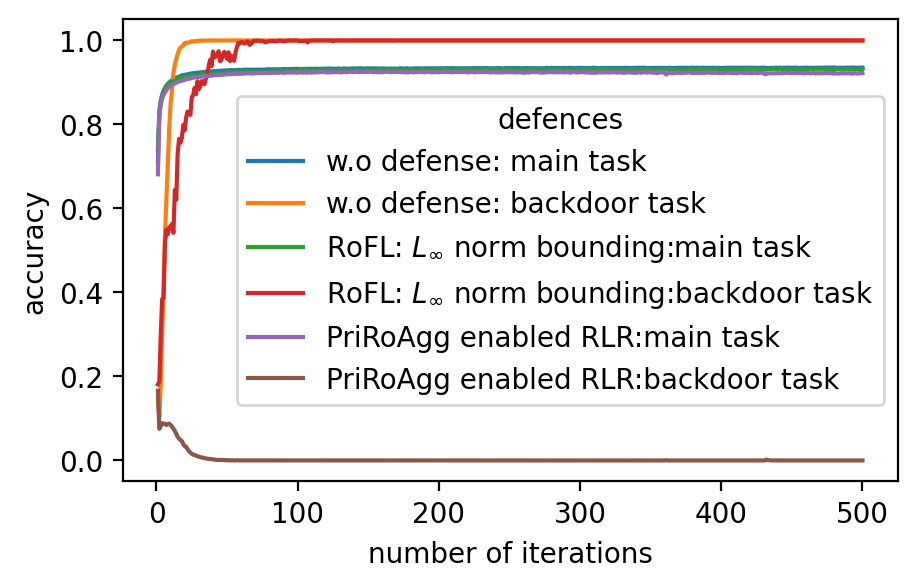}
	      \caption*{\small  (i) Backdoor on FMNIST}
	\end{minipage}
	\begin{minipage}{0.24\linewidth}
	    \includegraphics[width=\linewidth]{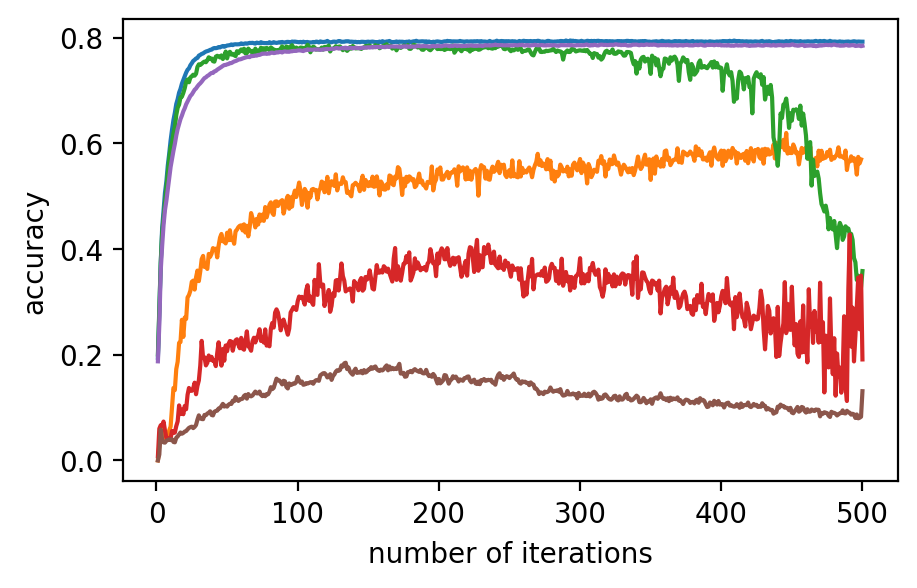}
	      \caption*{\small  (j) Backdoor on CIFAR10}
        \end{minipage} 
	\begin{minipage}{0.24\linewidth}
	    \includegraphics[width=\linewidth]{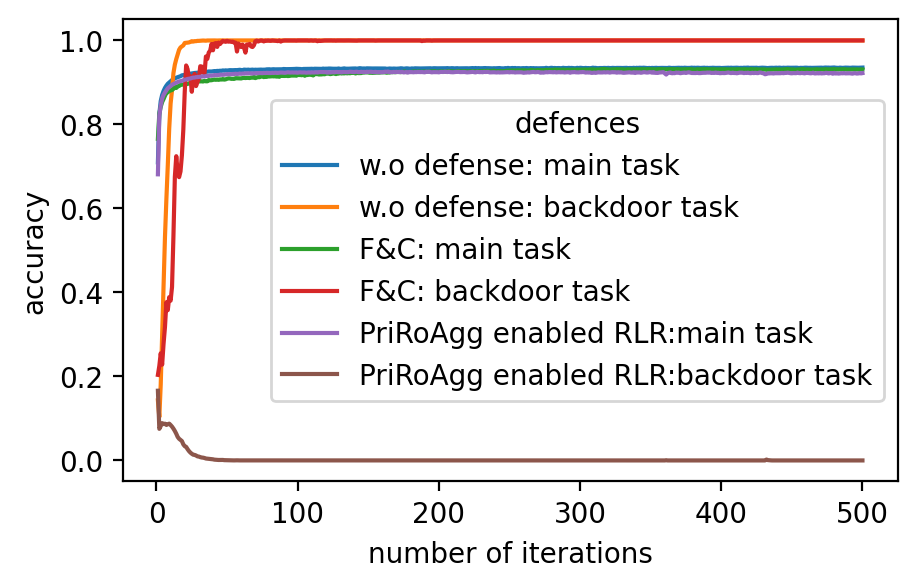}
	      \caption*{\small  (k) Backdoor on FMNIST}
         \end{minipage}
         \begin{minipage}{0.24\linewidth}
	    \includegraphics[width=\linewidth]{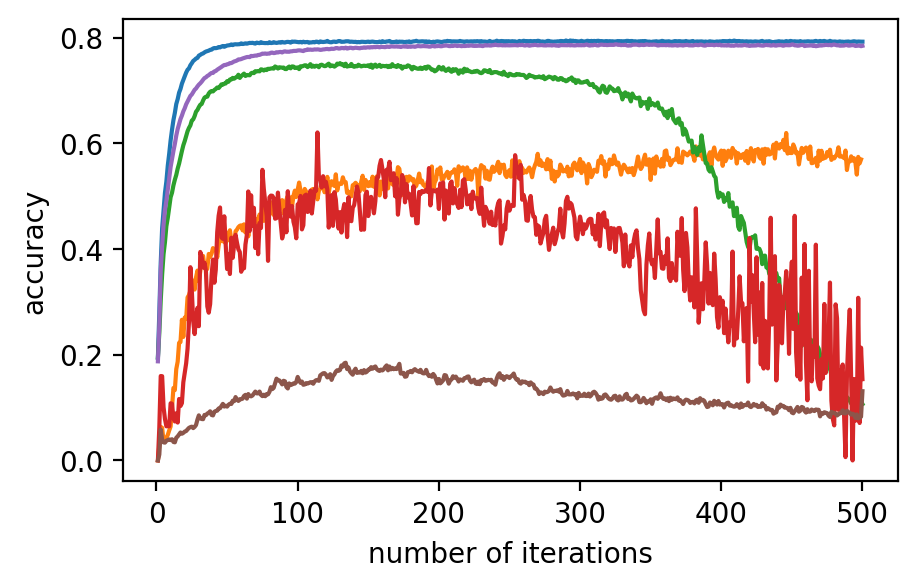}
	      \caption*{\small  (l) Backdoor on CIFAR10}
         \end{minipage}
\caption{\textbf{Learning dynamics} of different attacks and defenses on Fashion-MNIST and CIFAR10. (a)-(f)       are model poisoning attacks and (g)-(l) are backdoor attacks. Specifically, additive noise (a)(b), scaling (c)(d), Min-Max (e) and Min-Sum (f) attacks are demonstrated.    (g)(h), (i)(j) and (k)(l) show \scheme with RLR compared with EIFFeL, RoFL and EPPRFL, respectively.    Sub-figures with the same attack share the same legend.}
    \label{fig: attacks}
\vspace{-3mm}
\end{figure*}

\noindent\textbf{Baselines.} For defenses, we consider protocols that attempts to achieve robustness and privacy guarantees. BREA \cite{brea} and ByzSecAgg \cite{jahani2023byzantine} ensure robustness using the Byzantine-resilient algorithm, multi-Krum \cite{krum}.
RoFL \cite{lycklama2023rofl} and EIFFeL \cite{chowdhury2021eiffel}
utilize the validation functions to ensure the updates' robustness, where $L_2$ and $L_\infty$ norm bounding are used as the baselines.  We additionally use historical cosine similarity for EIFFeL's implementation to present the effect of different choices of validation functions. Note that RoFL and EIFFeL's validation functions rely on a public dataset and a given threshold, nevertheless, we allow their advantages with the dataset, which is isolated from the training dataset. For similarity and rank schemes in Trust-ZKFL~\cite{ma2024trusted}, they use Multi-Krum and SignGuard~\cite{xu2021signguard}, respectively. Consequently, similarity-based Trust-ZKFL corresponds to the Multi-Krum case mentioned before, while rank-based Trust-ZKFL is presented for comparison. 
For scarcity of backdoor-resilient privacy-preserving protocols in one-server setting, we consider another two-server protocol, EPPRFL, for backdoor attacks \cite{li2024efficiently}. It uses the F\&C algorithm adopted from \cite{karimireddy2021learning} to filter malicious updates based on historical updating information.

\noindent\textbf{Datasets}. We use Fashion-MNIST (FMNIST) \cite{xiao2017/online}, CIFAR-10 \cite{krizhevsky2009learning} 
to evaluate performance and test the end-to-end execution time. 
We assume the data is independent and identically distributed (i.i.d.)  unless otherwise specified.  

We choose $N=10$ for all Fashion-MNIST experiments and $N=40$ for all CIFAR-10 experiments, with $10\% - 30\%$ malicious users.     The security threshold is set to $T=\left \lceil 0.3N  \right \rceil $ and $K=\left \lfloor 0.7N \right \rfloor -1$ to maximize efficiency while ensuring compliance with privacy guarantee. Each experiment is run for $10$ times 
and the average results are reported. The number of global epoch is set to $500$.

\subsection{Results}

\begin{table*}[b]
    \centering
    \footnotesize
    \caption{Model Poisoning Performance}
    \label{table: performance poison}
    \begin{tabular}{ccccccccccccccc}\toprule
    & \multicolumn{7}{c}{Fashion-MNIST} & \multicolumn{7}{c}{CIFAR-10} \\
    \cmidrule(lr){2-8} \cmidrule(lr){9-15}
         & \hspace{-1mm}No Att.\hspace{-1mm} & \hspace{-1mm}Gaussian\hspace{-1mm} & \hspace{-1mm}Scale\hspace{-1mm} & \hspace{-1mm}MinMax\hspace{-1mm} & \hspace{-1mm}MinSum\hspace{-1mm} & \hspace{-1mm}PoisonFL\hspace{-1mm} & \hspace{-1mm}Sine\hspace{-1mm} & \hspace{-1mm}No Att.\hspace{-1mm} & \hspace{-1mm}Gaussian\hspace{-1mm} & \hspace{-1mm}Scale\hspace{-1mm} & \hspace{-1mm}MinMax\hspace{-1mm} & \hspace{-1mm}MinSum\hspace{-1mm} & \hspace{-1mm}PoisonFL\hspace{-1mm} & \hspace{-1mm}Sine\hspace{-1mm} 
         \\\midrule
         No Defence & \textbf{93.8}  & 10.0 & 10.0 & 10.0 & 10.0 & 10.0 & 10.0  & \textbf{79.5} & 10.0 & 10.0 & 10.0 & 10.0 & 10.0 & 10.0 \\
         Multi-Krum & - & 89.6 & 89.5 & 89.5 & 89.9 & 89.8 & 89.4 & - & 13.3 & 14.7 & 14.1 & 14.0 & 12.6 & 15.4 \\
         EIFFel ($L_2$) & - & 92.8  & \textbf{93.3} & 68.5 & \textbf{92.9} & \textbf{93.1} & 92.5 & - & 77.6 & 77.9 & 75.5 & 75.6 & 77.2 & 77.9\\
         EIFFel (cosine) & - & 10.0 & 10.0 & 68.9 & 92.6  & 10.0  & 10.0 & - & 10.0 & 10.0 & 10.0 & 10.0 & 10.0 & 10.0 \\
         RoFL ($L_\infty$) & - & \textbf{93.3} & \textbf{93.4} & 92.3 & 92.5 & \textbf{93.2} & 92.6 & - & 11.8 & 27.8 & 18.8 & 19.0 & 17.4 & 18.5\\
         SignGuard & - & \textbf{93.5} & 91.1 & 71.0 & 80.5& \textbf{93.3} & 10.0 & - & \textbf{79.5} & \textbf{78.8} &  10.0 & 11.3 & \textbf{78.7 }& 10.0 \\
         \scheme (ours) & - & \textbf{93.6}  & \textbf{93.6} & \textbf{92.6} & \textbf{93.1}  & 92.2 & \textbf{92.9} & - & \textbf{79.5} & \textbf{79.5} & \textbf{79.2} & \textbf{79.3} & \textbf{78.6} & \textbf{79.1}
        \\\bottomrule
    \end{tabular}
    \vspace{-3mm}
\end{table*}

\noindent\textbf{Performance}. 
We present the final accuracy for every attack-defense pair in Table ~\ref{table: performance poison}, ~\ref{table: performance backdoor}. We also notice the effectiveness of defenses changes throughout the training process, hence demonstrate the learning dynamics of selected experiments in Fig.~\ref{fig: attacks}.
The performance is measured by the testing accuracy of the global model and is presented as a percentage.
For model poisoning attacks, \scheme consistently yields robust results close to the no-attack baseline. The multi-Krum algorithm demonstrates convergence on FMNIST, although its accuracy is not as high as that of other convergent baselines. The performance of EIFFeL depends on the chosen validation predicate. On FMNIST, EIFFeL with cosine similarity poorly defends against all attacks except the Min-Sum attack, and EIFFeL with $L_2$-norm fails in the Min-Max attack.
RoFL generally performs well on FMNIST; however, it is no longer robust when using CIFAR-10. Trust-ZKFL with SignGuard experiences a performance degradation in omniscient attacks such as Min-Max, Min-Sum, and Sine attacks. As CIFAR-10 is a more complex dataset, multi-Krum, EIFFeL with cosine similarity, and RoFL fail to prevent model poisoning.
Instead of failing at the beginning of the training stage, RoFL experiences a drastic decrease around the $400$th iteration in Fig.~\ref{fig: attacks}(b). This behavior is pervasive when the learning dynamics change the distribution of updates. EIFFeL exhibits relatively good performance with a correct choice of validation predicate, but it falls short in terms of accuracy and stability compared to \scheme. Trust-ZKFL with SignGuard has comparable performance in simple attacks and PoisonFL attack, but \scheme consistently perform well in every attack. 

\begin{table}[t]
    \centering
    \caption{Backdoor Performance}
    \label{table: performance backdoor}
    \begin{tabular}{ccccccc}\toprule
    & \multicolumn{2}{c}{FMNIST} & \multicolumn{4}{c}{CIFAR10} \\
          & \multicolumn{2}{c}{Trojan A} & \multicolumn{2}{c}{Trojan B} & \multicolumn{2}{c}{IBA} \\
\cmidrule(lr){2-3} \cmidrule(lr){4-5} \cmidrule(lr){6-7} &  \textit{MA} $(\uparrow)$ & \textit{ASR} $(\downarrow)$ & \textit{MA} & \textit{ASR} & \textit{MA}  & \textit{ASR}  \\\midrule
         No Defence &  93.2 & 94.4 & 79.5 & 62.0 & 79.4 & 61.9 \\
         EIFFel ($L_2$) & 93.1 & 100.0 & 79.0 & 38.9 & 78.0 & 27.5\\
         EIFFel (cosine) & 93.4 & 100.0 & 50.1 & 40.7 & 22.1 & 23.0 \\
         RoFL ($L_\infty$) & 92.6  & 100.0 & 74.9 &42.7 & 59.2 & 25.4 \\
         F\&C & 93.3 & 100.0 & 50.5 & 46.2 &  37.4 & 24.3 \\
         \scheme (ours) & 92.2 & \textbf{0.0} & 78.7 & \textbf{13.1} &  78.9 & \textbf{12.4}
        \\\bottomrule
    \end{tabular}
    \vspace{-3mm}
\end{table}

For backdoor attacks, we evaluate the main task on the clean test dataset and the backdoor task on the poisoned test dataset. The learning dynamics are presented in (g)-(l) on Fig.~\ref{fig: attacks} and the final results are reported in Table~\ref{table: performance backdoor}. Generally, the compared defenses are effective in the early stages but fail in the long term.
As shown in (i) and (k), although RoFL and F\&C manage to slow down the implantation process on Fashion-MNIST, both ultimately succumb to the backdoor. Compared to them, \scheme with RLR performs significantly well against backdoor implantation on both datasets. It greatly supports the discussion in Section \ref{sec: challeges} about validation on the sole update not sufficiently supporting sophisticated model aggregation. The EIFFeL with cosine similarity has nearly no defensive capability against backdoor attacks, while EIFFeL and RoFL with norm bounding performs slightly better. However, they significantly reduces the main task performance when facing advanced attacks such as IBA. We consider multiple parameter settings for EPPRFL with F\&C and notice that F\&C highly relies on its early-stage updates' history; that is, if it cannot detect the malicious users in the first several rounds, it fails in the backdoor attack. Hence, F\&C exhibits notable instability shown in (l). 
In comparison, by allowing the overall statistics to aid aggregation, \scheme robustly defends against the backdoor in both backdoor attacks, while achieving the comparable performance in the main task as baseline.

\noindent \textbf{Complexity results}. We present the end-to-end runtime of \scheme with the instantiations of RLR and RFA. The instantiation with RFA is presented with the default setting of $N=10$ and $d=1000$. To provide a comprehensive and fair comparison, we implement EIFFeL and use identical encoding, decoding and ZKP algorithms across protocols. 
\begin{figure}[htb]
    \centering
    \includegraphics[width=\linewidth]{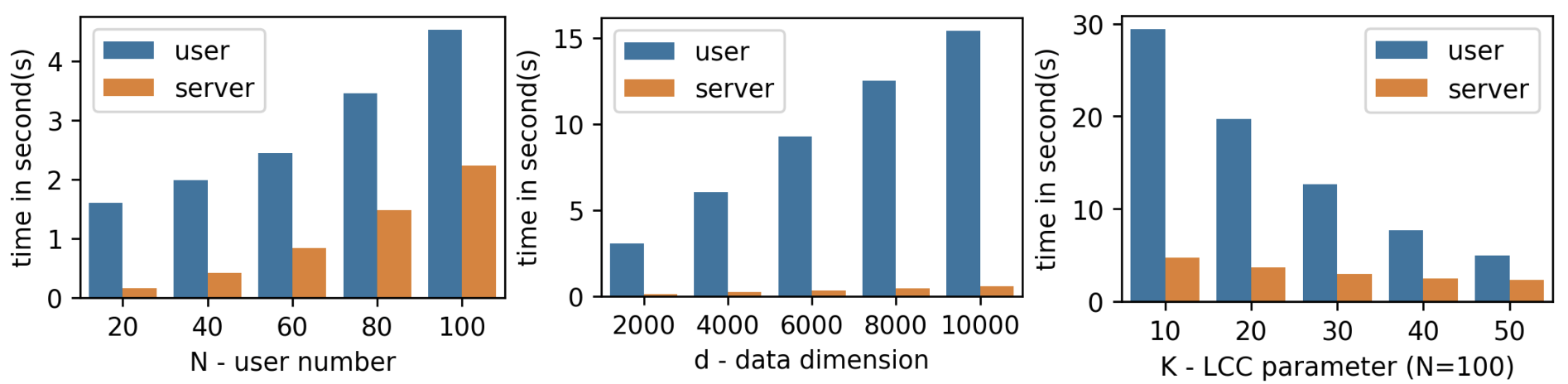}
    \caption{The user and the server overhead breakdown of \scheme per iteration. $N,d,K$ are the number of users, dimension of update, and LCC parameter, respectively.}
    \label{fig:overhead1}
    \vspace{-3mm}
\end{figure}
\begin{figure}[htb]
    \centering
    \includegraphics[width=\linewidth]{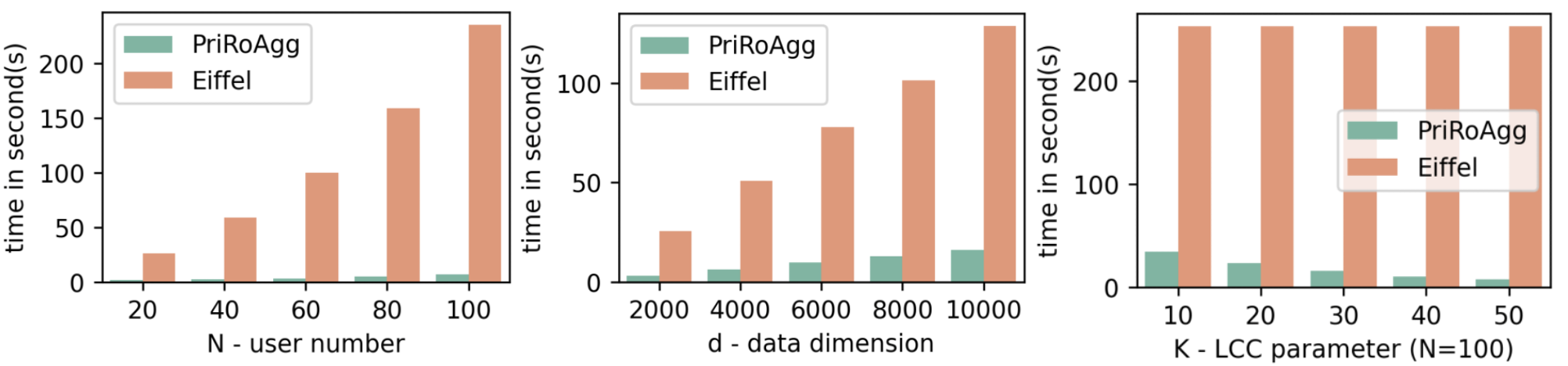}
    \caption{The end-to-end overhead of \scheme compared with EIFFeL. $N,d,K$ are the number of users, dimension of update, and LCC parameter, respectively.}
    \label{fig:overhead2}
\end{figure}

As shown in Fig.~\ref{fig:overhead1}, the time against the number of users and the dimension of update follows the complexity analysis in Table ~\ref{table: overhead}. The time cost increases linearly as it is primarily dominated by variables $d$ and $N$.For different LCC encoding parameter $K$ with fixed $T, N$ satisfying $K<N-T$, the larger LCC parameter $K$ shows higher encoding efficiency, where EIFFeL is approximately $K=1$ case.  In Fig.\ref{fig:overhead2}, compared with EIFFeL, \scheme benefits from higher encoding efficiency of LCC and repeated pattern circuit in improved SNIP. In summary, \scheme is superior in end-to-end efficiency, especially when $N$ is large allowing larger $K$ for LCC. This advantage is particularly evident in Fig.\ref{fig:overhead2}.

\section{Conclusion}
Under realistic consideration, FL requires protocols to provide both robustness and security at the same time. In this paper, we introduce a new notion of security for users' data privacy called aggregated privacy, 
and propose a novel framework \scheme to provide aggregated privacy to participants, while utilizing powerful robust algorithms to defend against both targeted and untargeted attacks. 
Two concrete protocols are proposed as instantiations of \scheme, with each addressing model poisoning and backdoor attacks. Security proofs and experiments are presented to demonstrate \scheme's superiority compared with previous works.

\section*{Acknowledgments}
This work is supported in part by the Fundamental Research Funds for the Central Universities (Grant No. 2242025K30025).
The work of T. Jahani-Nezhad and G. Caire was supported by the Gottfried Wilhelm Leibniz Prize 2021, awarded by the German Research Foundation (DFG).

\newpage

\appendices

\section{Extension for SNIP}\label{app:snip}

We introduce an extension of SNIP for repeated pattern circuits using LCC to reduce a term in the computational cost from $\mathcal{O}(d\log{d})$ to $\mathcal{O}(\frac{Nd}{K}\log{N})$. It is effective when $d>>N$ and $K$ is of order $\mathcal{O}(N)$, the term in \scheme becomes $\mathcal{O}(d\log{N})$. In both instantiations of \scheme, a special trait is that the users compute each entry of $\mathbf{x}_i$ in the same way. Specifically, RLR computes the $sign$ function for all entries of $\mathbf{x}_i$ and RFA computes the square for those. Utilizing such a trait, we split the repeated computations of $\mathbf{x}_i$ to computations of each coordinate of $\mathbf{x}_i$. 
All notations are directly inherited from Section.\ref{sec: SNIP}.

For demonstration, say we split $K$ arithmetic circuits $f^{(k)}, \forall k \in [K]$ from $f$ with the same topology structures. Each $f^{(k)}$ is applied on input $\mathbf{x}_i[k]$ and outputs $\mathbf{e}_i[k]$. For each $f^{(k)}(\mathbf{x}_i[k]) = \mathbf{e}_i[k]$, step 1 is executed in the same manner and results in $\mathbf{h}_i^{(k)}, k\in[K]$. Then, $\mathbf{h}_i^{(k)}, k\in[K]$ is viewed as the partitioned pieces from $\mathbb{F}_q^{2P-2}$ as in \eqref{encdoing goal}, on which \texttt{LCC.share} is applied with partitioning parameter $K$ instead of $1$. Therefore, the shares of $\mathbf{h}_i^{(k)}, k\in[K]$ are still $[\mathbf{h}_i]_j, \forall j\in[N]$. The rest of steps are performed identically except when the server reconstructs by \texttt{LCC.recon}, it should decode $K$'s many $0$ if the prover honestly executes the evaluation. 

\newpage

\section{Quantization technique}
\label{app: quantization}

In \scheme, we follow the same stochastic quantization method as ~\cite{brea}, which presents the theoretical guarantees of convergence. Given a real variable $x$ in range $(-1,1)$, it computes a probabilistic rounding $\mathcal{R}(x)$,
\begin{IEEEeqnarray}{c}\label{eq: quantization1}
    \mathcal{R}(x)=\left\{\begin{array}{@{}ll}
    \frac{\lfloor px\rfloor}{p},   &\mathrm{with~prob.}~1-(px-\lfloor px\rfloor)  \\
    \frac{\lfloor px\rfloor+1}{p},&\mathrm{with~prob.}~px-\lfloor px\rfloor
    \end{array}\right.
\end{IEEEeqnarray}
where $\lfloor x\rfloor$ denotes the largest integer no more than $x$, and $p$ is an integer that controls the quantization loss. 
Next, it maps the rounding to the finite field domain by $\mathcal{Q}$,
\begin{IEEEeqnarray}{c}\label{eq: quantization2}
    \mathcal{Q}(x)=\left\{\begin{array}{@{}ll}
    x,   &\mathrm{if}~x\geq 0  \\
    q+x,&\mathrm{if}~x<0
    \end{array}\right..
\end{IEEEeqnarray}
Then, we define the element-wise quantization function for \scheme to be: $\texttt{quantize}(\mathbf{x}) = \mathcal{Q}\Big(p\cdot\mathcal{R}(\mathbf{x})\Big)$. \texttt{dequantize($\cdot$)} is defined to be the inverse function that maps values in $\mathbb{F}_q$ to $\mathbb{R}$. 
In the experiment, $p$ is chosen between $10^3 - 10^4$ to achieve indistinguishable performance as real-domain execution.

\twocolumn

\end{document}